\documentclass[11pt]{llncs}
\usepackage{tabularx,booktabs,multirow,delarray,array}
\usepackage{graphicx,amsmath,amssymb}
\usepackage{latexsym}
\usepackage[linesnumbered, vlined, ruled]{algorithm2e}

\usepackage[in]{fullpage}
\def\calP{\mathcal{P}}
\def\bbR{\mathbb{R}}

\newtheorem{observation}{Observation}

\def\ftest0{{\em FTEST0}}
\def\dftest0{{\em DFTEST0}}
\def\msearch{{\em MSEARCH}}
\def\calA{\mathcal{A}}
\def\calM{\mathcal{M}}
\def\ftestnew{{\em FTEST1}}
\def\dftestnew{{\em DFTEST1}}

\begin{document}

\title{An $O(n\log n)$-Time Algorithm for the $k$-Center Problem in Trees\thanks{A preliminary version of this paper will appear in the Proceedings of the 34th International Symposium on Computational Geometry (SoCG 2018).}
}

\author{Haitao Wang\inst{1}\and Jingru Zhang\inst{2}}

\institute{Department of Computer Science\\
Utah State University, Logan, UT 84322, USA\\
\email{haitao.wang@usu.edu} \and
Department of Computer Science\\ 
Marshall University, Huntington, WV 25755, USA\\
\email{jingru.zhang@marshall.edu}
}

\maketitle

\pagestyle{plain}
\pagenumbering{arabic}
\setcounter{page}{1}

\begin{abstract}
We consider a classical $k$-center problem in trees. Let $T$ be a tree
of $n$ vertices and every vertex has a nonnegative weight. The problem
is to find $k$ centers on the edges of $T$ such that the maximum
weighted distance from all vertices to their closest centers is
minimized. Megiddo and Tamir (SIAM J. Comput., 1983) gave an algorithm
that can solve the problem in $O(n\log^2 n)$ time by using Cole's parametric search.
Since then it has been open for over three decades whether the
problem can be solved in $O(n\log n)$ time. In this paper, we present
an $O(n\log n)$ time algorithm for the problem and thus settle the
open problem affirmatively.
\end{abstract}

%\smallskip
%\noindent \textbf{Keywords:} facility locations, $k$-center, uncertain data, tree networks, algorithms and data structures, computational geometry

\section{Introduction}
\label{sec:Gintro}

In this paper, we study a classical $k$-center problem in trees. Let
$T$ be a tree of $n$ vertices. Each edge $e(u,v)$ connecting two
vertices $u$ and $v$ has a positive length $d(u,v)$, and we consider
the edge as a line segment of length $d(u,v)$ so that we can talk
about ``points'' on the edge. For any two points $p$ and $q$ of $T$,
there is a unique path in $T$ from $p$ to $q$, denoted by $\pi(p,q)$,
and by slightly abusing the notation, we use $d(p,q)$ to denote the
length of $\pi(p,q)$. Each vertex $v$ of $T$ is associated with a weight $w(v)\geq
0$. The {\em $k$-center problem} is to compute a set $Q$ of $k$ points
on $T$, called {\em centers}, such that the maximum weighted distance
from all vertices of $T$ to their closest centers is minimized,
or formally, the value $\max_{v\in V(T)}\min_{q\in Q}\{w(v)\cdot
d(v,q)\}$ is minimized, where $V(T)$ is the vertex set of $T$. Note
that each center can be in the interior of an edge of $T$.
%The vertices of $T$ are also called {\em demand points}.

Kariv and Hakimi~\cite{ref:KarivAn79} first gave an $O(n^2\log n)$
time algorithm for the problem. Jeger and Kariv~\cite{ref:JegerAl85}
proposed an $O(kn\log n)$ time algorithm. Megiddo and
Tamir~\cite{ref:MegiddoNe83} solved the problem in $O(n\log^2
n\log\log n)$ time, and the running time of their algorithm can be
reduced to
$O(n\log^2 n)$ by applying Cole's parametric
search~\cite{ref:ColeSl87}. Some progress has been made very recently
by Banik et al.~\cite{ref:BanikCo16} for small values of $k$, where an $O(n\log n+k\log^2
n\log(n/k))$-time algorithm and another $O(n\log n+k^2\log^2
(n/k))$-time algorithm were given.

Since Megiddo and Tamir's work~\cite{ref:MegiddoNe83}, it has been
open whether the problem can be solved in $O(n\log n)$ time. In this
paper, we settle this three-decade long open problem affirmatively by
presenting an $O(n\log n)$-time algorithm. Note that the previous
$O(n\log^2 n)$-time algorithm~\cite{ref:ColeSl87,ref:MegiddoNe83} and
the first algorithm in~\cite{ref:BanikCo16} both rely on Cole's
parametric search, which involves a large constant in the time complexity due to the AKS sorting
network~\cite{ref:AjtaiAn83}. Our algorithm, however, avoids Cole's
parametric search.

If each center is required to be located at a vertex of $T$, then we call it the {\em discrete} case. The previously best-known algorithm for this case runs in $O(n\log^2 n)$ time~\cite{ref:MegiddoAn81}. Our techniques also solve the discrete case in $O(n\log n)$ time.

\subsection{Related Work}

Many variations of the $k$-center problem have been studied.
If $k=1$, then the problem is solvable in $O(n)$
time~\cite{ref:MegiddoLi83}. If $T$ is a path, the $k$-center
problem was already solved in $O(n\log n)$
time~\cite{ref:ChenEf15,ref:ColeSl87,ref:MegiddoNe83}, and
Bhattacharya and Shi~\cite{ref:BhattacharyaOp07} also gave an
algorithm whose running time is linear in $n$ but exponential in $k$.
%If we require the centers to be located at the vertices of $T$, then
%the problem was solved in $O(n\log^2 n)$ time~\cite{ref:MegiddoAn81}.

For the unweighted case where the vertices of $T$ have the same weight,
an $O(n^2\log n)$-time algorithm was given
in~\cite{ref:ChandrasekaranPo82} for the $k$-center problem. Later,
Megiddo et al.~\cite{ref:MegiddoAn81} solved the problem in $O(n\log^2
n)$ time, and the algorithm was improved to $O(n\log n)$
time~\cite{ref:FredericksonFi83}. Finally,
Frederickson~\cite{ref:FredericksonPa91} solved the problem in $O(n)$
time. The above four papers also solve the discrete case and the following problem
version in the same running times:
%In the first version, it is required that each center be located at a vertex of $T$. In the second version,
All points of $T$ are considered as demand points and the
centers are required to be at vertices of $T$. Further, if all points of
$T$ are demand points and centers can be any points of $T$,  Megiddo
and Tamir solved the problem in $O(n\log^3 n)$
time~\cite{ref:MegiddoNe83}, and the running time can be reduced to
$O(n\log^2 n)$ by applying Cole's parametric
search~\cite{ref:ColeSl87}.

As related problems, Frederickson~\cite{ref:FredericksonOp91} presented $O(n)$-time algorithms for the following tree partitioning problems: remove $k$ edges from $T$ such that the maximum (resp., minimum) total weight of all connected subtrees is minimized (resp., maximized).

Finding $k$ centers in a general graph is NP-hard~\cite{ref:KarivAn79}. The geometric version of the problem in the plane is also NP-hard~\cite{ref:MegiddoOn84}, i.e., finding $k$ centers for $n$ demanding points in the plane.
Some special cases, however, are solvable in polynomial time. For example, if $k=1$, then the problem can be solved in $O(n)$ time~\cite{ref:MegiddoLi83}, and if $k=2$, it can be solved in $O(n\log^2 n\log^2 \log n)$ time~\cite{ref:ChanMo99} (also refer to~\cite{ref:AgarwalAn08} for a faster randomized algorithm). If we require all centers to be on a given line, then the problem of finding $k$ centers can be solved in polynomial time~\cite{ref:BrassTh11,ref:KarmakarSo13,ref:WangLi16}. Recently, problems on uncertain data have been studied extensively and some $k$-center problem variations on uncertain data were also considered, e.g., \cite{ref:CormodeAp08,ref:HuangSt17,ref:WangCo17,ref:WangCo17WADS,ref:WangOn15,ref:WangA16}.

\subsection{Our Approach}
\label{sec:approach}

We discuss our approach for the non-discrete problem, and that for the discrete case is similar (and even simpler).
Let $\lambda^*$ be the optimal objective value, i.e., $\lambda^*=\max_{v\in V(T)}\min_{q\in Q}\{w(v)\cdot d(v,q)\}$ for an optimal solution $Q$. A {\em feasibility test} is to determine whether $\lambda\geq \lambda^*$ for a given value $\lambda$, and if yes, we call $\lambda$ a {\em feasible value}. Given any $\lambda$, the feasibility test can be done in $O(n)$ time~\cite{ref:KarivAn79}.

Our algorithm follows an algorithmic scheme in~\cite{ref:FredericksonPa91} for the unweighted case, which is similar to that in~\cite{ref:FredericksonOp91} for the tree partition problems. However, a big difference is that three schemes were proposed in~\cite{ref:FredericksonOp91,ref:FredericksonPa91} to gradually solve the problems in $O(n)$ time, while our approach only follows the first scheme and this significantly simplifies the algorithm. %comparing with those in~\cite{ref:FredericksonOp91,ref:FredericksonPa91}.
One reason the first scheme is sufficient to us is that our algorithm runs in $O(n\log n)$ time, which has a logarithmic factor more than the feasibility test algorithm. In contrast, most efforts of the last two schemes of~\cite{ref:FredericksonOp91,ref:FredericksonPa91} are to reduce the running time of the algorithms to $O(n)$, which is the same as their corresponding feasibility test algorithms.

More specifically, our algorithm consists of two phases. The first phase will gather information so that each feasibility test can be done faster in sub-linear time. By using the faster feasibility test algorithm, the second phase computes the optimal objective value $\lambda^*$. As in~\cite{ref:FredericksonPa91}, we also use a {stem-partition} of the tree $T$.
%which generalizes the path-partition in~\cite{ref:FredericksonOp91}.
In addition to a matrix searching algorithm~\cite{ref:FredericksonGe84}, we utilize some other techniques, such as { the 2D sublist LP queries}~\cite{ref:ChenAp13} and line arrangements searching~\cite{ref:ChenA13}, etc.

%\vspace{-0.15in}
\paragraph{Remark.} It might be tempting to see whether the techniques of~\cite{ref:FredericksonOp91,ref:FredericksonPa91} can be adapted to solving the problem in $O(n)$ time. Unfortunately, we found several obstacles that prevent us from doing so. For example, a key ingredient of the techniques in~\cite{ref:FredericksonOp91,ref:FredericksonPa91} is to build sorted matrices implicitly in $O(n)$ time so that each matrix element can be obtained in $O(1)$ time. In our problem, however, since the vertices of the tree $T$ have weights, it is elusive how to achieve the same goal. Indeed, this is also one main difficulty to solve the problem even in $O(n\log n)$ time (and thus makes the problem open for such a long time). As will be seen later, our $O(n\log n)$ time algorithm circumvents the difficulty by combining several techniques. Based on our study, although we do not have a proof, we suspect that $\Omega(n\log n)$ is a lower bound of the problem.
%We leave it as an open problem.
%\footnote{We feel that the best we can hope is to obtain each such matrix element in $O(\log n)$ time even if we allow $O(n\log n)$ time preprocessing.}.
\vspace{0.1in}

The rest of the paper is organized as follows. In Section~\ref{sec:pre}, we review some previous techniques that will be used later. In Section~\ref{sec:stem}, we describe our techniques for dealing with a so-called ``stem''. We finally solve the $k$-center problem on $T$ in Section~\ref{sec:tree}. By slightly modifying the techniques, we solve the discrete case in Section~\ref{sec:discrete}.
%Section~\ref{sec:conclude} concludes the paper.

\section{Preliminaries}
\label{sec:pre}

In this section we review some techniques that will be used later in our algorithm.

\subsection{The Feasibility Test \ftest0}
\label{sec:test}
Given any value $\lambda$, the feasibility test is to determine whether $\lambda$ is feasible, i.e., whether $\lambda\geq\lambda^*$. We say that a vertex $v$ of $T$ is {\em covered} (under $\lambda$) by a center $q$ if $w(v)\cdot d(v,q)\leq \lambda$. Note that $\lambda$ is feasible if and only if we can place $k$ centers in $T$ such that all vertices are covered. In the following we describe a linear-time feasibility test algorithm, which is essentially the same as the one in~\cite{ref:KarivAn79} although our description is much simpler.

We pick a vertex of $T$ as the root, denoted by $\gamma$. For each vertex $v$, we use $T(v)$ to denote the subtree of $T$ rooted at $v$. Following a post-order traversal on $T$, we place centers in a bottom-up and greedy manner.
For each vertex $v$, we maintain two values $sup(v)$ and $dem(v)$, where $sup(v)$ is the distance from $v$ to the closest center that has been placed in $T(v)$, and $dem(v)$ is the maximum distance from $v$ such that if we place a center $q$
%outside $T(v)\setminus \{v\}$
within such a distance from $v$ then all uncovered vertices of $T(v)$ can be covered by $q$. We also maintain a variable $count$ to record the number of centers that have been placed so far. Refer to Algorithm~\ref{algo:ftest0} for the pseudocode.

Initially, $count=0$, and for each vertex $v$, $sup(v)=\infty$ and $dem(v)=\frac{\lambda}{w(v)}$. Following a post-order traversal on $T$, suppose vertex $v$ is being visited. For each child $u$ of $v$, we update $sup(v)$ and $dem(v)$ as follows. If $sup(u)\leq dem(u)$, then we can use the center of $T(u)$ closest to $u$ to cover the uncovered vertices of $T(u)$, and thus we reset $sup(v)=\min\{sup(v),sup(u)+d(u,v)\}$.
Note that since $u$ connects $v$ by an edge, $d(v,u)$ is the length of the edge.
Otherwise, if $dem(u)<d(u,v)$, then we place a center on the edge $e(u,v)$ at distance $dem(u)$ from $u$, so we update $count=count+1$ and $sup(v)=\min\{sup(v),d(u,v)-dem(u)\}$. Otherwise (i.e., $dem(u)\geq d(u,v)$), we update $dem(v)=\min\{dem(v),dem(u)-d(u,v)\}$.

After the root $\gamma$ is visited, if $sup(\gamma)> dem(\gamma)$, then we place a center at $\gamma$ and update $count=count+1$. Finally, $\lambda$ is feasible if and only if $count\leq k$. The algorithm runs in $O(n)$ time. We use \ftest0\ to refer to the algorithm.

\paragraph{Remark.} The
algorithm \ftest0\ actually partitions $T$ into at most $k$ disjoint connected subtrees such that the vertices in
each subtree is covered by the same center that is located in the subtree. We will make use of this observation later.
\vspace{0.1in}

To solve the $k$-center problem, the key is to compute $\lambda^*$, after which we can find $k$ centers by applying \ftest0\ with $\lambda=\lambda^*$.

\begin{algorithm}[h]
\caption{The feasibility test algorithm \ftest0}
\label{algo:ftest0}
%\SetAlgoNoLine
%\SetAlgoLined
\KwIn{The tree $T$ with root $\gamma$ and a value $\lambda$}
\KwOut{Determine whether $\lambda$ is feasible}
\BlankLine
%\tcc{All indices below are understood modulo $n$.}
$count\leftarrow 0$\;
\For{each vertex $v$}
{
  $sup(v)\leftarrow \infty$, $dem(v)\leftarrow \frac{\lambda}{w(v)}$\;
}
\For{\em each vertex $v$ in the post-order traversal of $T$}
{
    \For{\em each child $u$ of $v$}
	{
     \eIf{$sup(u)\leq dem(u)$}
	{
		$sup(v)=\min\{sup(v),sup(u)+d(u,v)\}$\;
    }
	{
		\eIf{$dem(u)< d(u,v)$}
		{
           $count++$\tcc*{place a center on the edge $e(u,v)$ at distance $dem(u)$ from $u$}
		   $sup(v)=\min\{sup(v),d(u,v)-dem(u)\}$\;
\label{ln:ftest0}
	        }
            {
               $dem(v)=\min\{dem(v),dem(u)-d(u,v)\}$\;
            }
		
		}
}
}
\If{$sup(\gamma)>dem(\gamma)$}
{
        $count++$\tcc*{place a center at the root $\gamma$}
}
%\eIf{$count\leq k$}
%{
%        return true\; %\tcc*{$\gamma$ is feasible}
%}
%{
%        return false\; %\tcc*{$\gamma$ is not feasible}
%}
Return true if and only if $count\leq k$\;
\end{algorithm}

\subsection{A Matrix Searching Algorithm}
We review an algorithm \msearch, which was proposed in~\cite{ref:FredericksonGe84} and was widely used, e.g.,~\cite{ref:FredericksonOp91,ref:FredericksonPa91,ref:FredericksonFi83}.
A matrix is {\em sorted} if elements in every row and every column are in nonincreasing order. Given a set of sorted matrices, a searching range $(\lambda_1,\lambda_2)$ such that $\lambda_2$ is feasible and $\lambda_1$ is not, and a stopping count $c$, \msearch\ will produce a sequence of values one at a time for feasibility tests, and after each test, some elements in the matrices will be discarded. Suppose a value $\lambda$ is produced. If $\lambda\not\in (\lambda_1,\lambda_2)$, we do not need to test $\lambda$. If $\lambda$ is feasible, then $\lambda_2$ is updated to $\lambda$; otherwise, $\lambda_1$ is updated to $\lambda$. \msearch\ will stop once the number of remaining elements in all matrices is at most $c$. Lemma~\ref{lem:msearch} is proved in~\cite{ref:FredericksonGe84} and we slightly change the statement to accommodate our need.

\begin{lemma}\label{lem:msearch} {\em \cite{ref:FredericksonGe84,ref:FredericksonOp91,ref:FredericksonPa91,ref:FredericksonFi83}}
Let $\calM$ be a set of $N$ sorted matrices $\{M_1, M_2, \ldots, M_N\}$ such that $M_j$ is of dimension $m_j\times n_j$ with $m_j\leq n_j$, and $\sum_{j=1}^{N}m_j=m$. Let $c\geq 0$. The number of feasibility tests needed by \msearch\ to discard
all but at most $c$ of the elements is $O(\max\{\log\max_j\{n_j\}, \log(\frac{m}{c+1})\})$,
and the total time of \msearch\ exclusive of feasibility tests is $O(\kappa\cdot \sum_{j=1}^{N}m_j\log (\frac{2n_j}{m_j}))$, where $O(\kappa)$ is the time for evaluating each matrix element (i.e., the number of matrix elements that need to be evaluated is $O(\sum_{j=1}^{N}m_j\log (\frac{2n_j}{m_j}))$).
\end{lemma}

\subsection{The 2D Sublist LP Queries}

Let $H=\{h_1,h_2,\ldots,h_m\}$ be a set of $m$ upper half-planes in the plane. Given two indices $i$ and $j$ with $1\leq i\leq j\leq m$, a {\em 2D sublist LP query} asks for the lowest point in the common intersection of $h_i,h_{i+1},\ldots,h_j$. The {\em line-constrained} version of the query is: Given a vertical line $l$ and two indices $i$ and $j$ with $1\leq i\leq j\leq m$, the query asks for the lowest point on $l$ in the  common intersection of $h_i,h_{i+1},\ldots,h_j$.
Lemma~\ref{lem:sublist} was proved in~\cite{ref:ChenAp13} (i.e., Lemma 8 and the discussion after it; the query algorithm for the line-constrained version is used as a procedure in the proof of Lemma 8).

\begin{lemma}\label{lem:sublist}{\em \cite{ref:ChenAp13}}
%Let $H=\{h_1,h_2,\ldots,h_m\}$ be a set of $m$ upper half-planes.
We can build a data structure for $H$ in $O(m\log m)$ time such that each 2D sublist LP query can be answered in $O(\log^2 m)$ time, and each line-constrained query can be answered in $O(\log m)$ time.
\end{lemma}

\paragraph{Remark.} With $O(m\log m)$ preprocessing time, any $\text{poly}(\log m)$-time algorithms for both query problems would be sufficient for our purpose, where $\text{poly}(\cdot)$ is any polynomial function.

\subsection{Line Arrangement Searching}

Let $L$ be a set of $m$ lines in the plane. Denote by $\calA(L)$ the arrangement of the lines of $L$, and let $y(v)$ denote the $y$-coordinate of each vertex $v$ of $\calA(L)$. Let $v_1(L)$ be the lowest vertex of $\calA(L)$ whose $y$-coordinate is a feasible value, and let $v_2(L)$ be the highest vertex of $\calA(L)$ whose $y$-coordinate is smaller than that of $v_1(L)$. By their definitions, $y(v_2(L))<\lambda^*\leq y(v_1(L))$, and $\calA(L)$ does not have a vertex $v$ with $y(v_2(L))<y(v)< y(v_1(L))$. Lemma~\ref{lem:arrangement} was proved in~\cite{ref:ChenA13}.

\begin{lemma}\label{lem:arrangement}{\em \cite{ref:ChenA13}}
Both vertices  $v_1(L)$ and  $v_2(L)$ can be computed in $O((m+\tau)\log m)$ time, where $\tau$ is the time for a feasibility test.
\end{lemma}

\paragraph{Remark.}
Alternatively, we can use Cole's parametric search~\cite{ref:ColeSl87} to compute the two vertices. First, we sort the lines of $L$ by their intersections with the horizontal line $y=\lambda^*$, and this can be done in $O((m+\tau)\log m)$ time by Cole's parametric search~\cite{ref:ColeSl87}. Then, $v_1(L)$ and  $v_2(L)$ can be found in additional $O(m)$ time because each of them is an intersection of two adjacent lines in the above sorted order. The line arrangement searching technique, which modified the slope selection algorithms~\cite{ref:BronnimannOp98,ref:KatzOp93}, avoids Cole's parametric search~\cite{ref:ColeSl87}.
\vspace{0.1in}

In the following, we often talk about some problems in the plane $\bbR^2$, and
if the context is clear, for any point $p\in \bbR^2$, we use $x(p)$ and $y(p)$ to denote its $x$- and $y$-coordinates, respectively.

\section{The Algorithms for Stems}
\label{sec:stem}

In this section, we first define {\em stems}, which are similar in spirit to those proposed in \cite{ref:FredericksonPa91} for the unweighted case. Then, we will present two algorithms for solving the $k$-center problem on a stem, and both techniques will be used later for solving the problem in the tree $T$.

Let $\widehat{P}$ be a path of $m$ vertices, denoted by $v_1,v_2,\ldots,v_m$, sorted from left to right. For each vertex $v_i$, other than its incident edges in $\widehat{P}$, $v_i$ has at most two additional edges connecting two vertices $u_i$ and $w_i$ that are not in $\widehat{P}$. Either vertex may not exist. Let $P$ denote the union of $\widehat{P}$ and the above additional edges (e.g., see Fig.~\ref{fig:stem}). For any two points $p$ and $q$ on $P$, we still use $\pi(p,q)$ to denote the unique path between $p$ and $q$ in $P$, and use $d(p,q)$ to denote the length of the path.
%Similarly, we use $\lambda^*$ to denote the optimal objective value for the $k$-center problem on $P$.
With respect to a range $(\lambda_1,\lambda_2)$, we call $P$ a {\em stem} if the following holds:
%Suppose we have a range $(\lambda_1,\lambda_2]$ that contains $\lambda^*$.
For each $i\in [1,m]$, if $u_i$ exists, then $w(u_i)\cdot d(u_i,v_i)\leq \lambda_1$; if $w_i$ exists, then $w(w_i)\cdot d(w_i,v_i)\geq \lambda_2$.
%which implies that any center on $P'$ cannot cover $w_i$ and we have to place a center on the edge $e(w_i,v_i)\setminus\{v_i\}$ to cover $w_i$.

\begin{figure}[t]
\begin{minipage}[t]{\linewidth}
\begin{center}
\includegraphics[totalheight=1.0in]{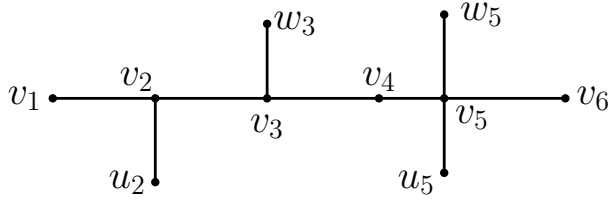}
\caption{\footnotesize Illustrating a stem.}
\label{fig:stem}
\end{center}
\end{minipage}
\vspace*{-0.15in}
\end{figure}

Following the terminology in \cite{ref:FredericksonPa91}, we call $e(v_i,u_i)$ a {\em thorn} and $e(v_i,w_i)$ a {\em twig}. Each $u_i$ is called a {\em thorn vertex} and each $w_i$ is called a {\em twig vertex}. $\widehat{P}$ is called the {\em backbone} of $P$, and each vertex of $\widehat{P}$ is called a {\em backbone vertex}.
We define $m$ as the {\em length} of $P$.
The total number of vertices of $P$ is at most $3m$.

\paragraph{Remark.}
Our algorithm in Section~\ref{sec:tree} will produce stems $P$ as defined above, where $\widehat{P}$ is a path in $T$ and all vertices of $P$ are also vertices of $T$. However, each thorn $e(u_i,v_i)$ may not be an original edge of $T$, but it corresponds to the path between $u_i$ and $v_i$ in $T$ in the sense that the length of $e(u_i,v_i)$ is equal to the distance between $u_i$ and $v_i$ in $T$. This is also the case for each twig $e(w_i,v_i)$.
Our algorithm in Section~\ref{sec:tree} will maintain a range $(\lambda_1,\lambda_2)$ such that $\lambda_1$ is not feasible and $\lambda_2$ is feasible, i.e., $\lambda^*\in (\lambda_1,\lambda_2]$. Since any feasibility test will be made to a value $\lambda\in (\lambda_1,\lambda_2)$, the above definitions on thorns and twigs imply the following: For each thorn vertex $u_i$, we can place a center on the backbone to cover it (under $\lambda$), and for each twig vertex $w_i$, we need to place a center on  the edge $e(w_i,v_i)\setminus\{v_i\}$ to cover it.
\vspace{0.1in}

In the sequel we give two different techniques for solving the $k$-center problem on the stem $P$. In fact, in our algorithm for the $k$-center problem on $T$ in Section~\ref{sec:tree}, we use these techniques to process a stem $P$, rather than directly solve the $k$-center problem on $P$. Let $\lambda^*$ temporarily refer to the optimal objective value of the $k$-center problem on $P$ in the rest of this section, and we assume $\lambda^*\in (\lambda_1,\lambda_2]$.
%Both algorithms are more or less based on the following observation.

%\begin{lemma}\label{lem:opt}
%There exists an optimal solution for the $k$-center problem on $P$ in which there is a center $q$ such that $w(v)\cdot d(v,q)=w(v')\cdot d(v',q)=\lambda^*$ and $q$ covers all vertices in the path $\pi(v,v')$ for two vertices $v$ and $v'$ of $P$. This applies to the $k$-center problem on the tree $T$ as well.
%\end{lemma}
%\begin{proof}
%In the following, we prove the lemma on $P$ and the proof is applicable to $P=T$ as well.
%
%Suppose we apply our feasibility test algorithm on $\lambda=\lambda^*$ and the stem $P$. Then, based on our algorithm, we will find a set of at most $k$ centers. In fact, the algorithm actually partitions $P$ into connected subtrees such that the vertices in each subtree is covered by a center. Further, there must be a subtree $P''$ that has a center $q$ and two vertices $v'$ and $v$ such that $w(v)\cdot d(v,q)=w(v')\cdot d(v',q)=\lambda^*$, since otherwise we would adjust the positions of the centers such that the maximum weighted distance from all vertices of $P$ to their closest centers is strictly smaller than $\lambda^*$. Since $P''$ is connected and both $v$ and $v'$ are in $P''$, $\pi(v,v')$ is also in $P''$. Thus, all vertices of $P''$ are covered by $q$. This proves the lemma.  \qed
%\end{proof}

\subsection{The First Algorithm}
\label{sec:first}

This algorithm is motivated by the following easy observation: there exist two vertices $v$ and $v'$ in $P$ such that a center $q$ is located in the path $\pi(v,v')$ and $w(v)\cdot d(q,v)=w(v')\cdot d(q,v')=\lambda^*$ (since otherwise we could move $q$ on $P$ to achieve such a situation).

We assume that all backbone vertices of $P$ are in the $x$-axis of an $xy$-coordinate system $\bbR^2$ where $v_1$ is at the origin and each $v_i$ has $x$-coordinate $d(v_1,v_i)$. Each $v_i$ defines two lines $l^+(v_i)$ and $l^-(v_i)$ both containing $v_i$ and with slopes $w(v_i)$ and $-w(v_i)$, respectively  (e.g., see Fig.~\ref{fig:lines}).
Each thorn $u_i$ also defines two lines $l^+(u_i)$ and $l^-(u_i)$ as follows. Define $u_i^l$ (resp., $u_i^r$) to be the point in $\bbR^2$ on the $x$-axis with $x$-coordinate $d(v_1,v_i)-d(u_i,v_i)$ (resp., $d(v_1,v_i)+d(u_i,v_i)$). Hence, $u_i^l$ (resp., $u_i^r$) is to the left (resp., right) of $v_i$ with distance $d(u_i,v_i)$ from $v_i$. Define $l^+(u_i)$ to be the line through $u_i^l$ with slope $w(u_i)$ and $l^-(u_i)$
to be the line through $u_i^r$ with slope $-w(u_i)$. Note that $l^+(u_i)$ and $l^-(u_i)$ intersect at the point whose $x$-coordinate is the same as that of $v_i$ and whose $y$-coordinate is equal to $w(u_i)\cdot d(u_i,v_i)$. For each twig vertex $w_i$, we define points $w_i^l$ and $w_i^r$, and lines $l^+(w_i)$ and $l^-(w_i)$, in the same way as those for $u_i$.

Consider a point $q$ on the backbone of $P$ to the right side of $v_i$. It can be verified that the weighted distance $w(v_i)\cdot d(v_i,q)$ from $v_i$ to $q$ is exactly equal to the $y$-coordinate of the intersection between $l^+(v_i)$ and the vertical line through $q$. If $q$ is on the left side of $v_i$, we have a similar observation for $l^-(v_i)$.  This is also true for $u_i$ and $w_i$.

\begin{figure}[t]
\begin{minipage}[t]{\linewidth}
\begin{center}
\includegraphics[totalheight=1.3in]{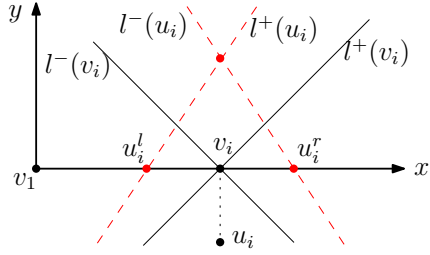}
\caption{\footnotesize Illustrating the definitions of the lines defined by a backbone vertex $v_i$ and its thorn vertex $u_i$.}
\label{fig:lines}
\end{center}
\end{minipage}
\vspace*{-0.15in}
\end{figure}

Let $L$ denote the set of the lines in $\bbR^2$ defined by all vertices of $P$. Note that $|L|\leq 6m$. Based on the above observation, $\lambda^*$ is equal to the $y$-coordinate of a vertex of the line arrangement $\calA(L)$ of $L$. More precisely, $\lambda^*$ is equal to the $y$-coordinate of the vertex $v_1(L)$, as defined in Section~\ref{sec:pre}. By Lemma~\ref{lem:arrangement}, we can compute $\lambda^*$ in $O((m+\tau)\log m)$ time.

\subsection{The Second Algorithm}
\label{sec:second}

%This algorithm is motivated by the following observation: if we apply the feasibility test algorithm \ftest0\ on $\lambda=\lambda^*$ and $P$, \ftest0\ will find a set of centers such that each center covers vertices in a connected subtree of $P$.

This algorithm relies on the algorithm \msearch. We first form a set of sorted matrices.

For each $i\in [1,m]$, we define the two lines $l_i^+(v_i)$ and $l_i^-(v_i)$ in $\bbR^2$ as above in Section~\ref{sec:first}. If $u_i$ exists, then we also define $l_i^+(u_i)$ and $l_i^-(u_i)$ as before; otherwise, both $l_i^+(u_i)$ and $l_i^-(u_i)$ refer to the $x$-axis.
Let $h_{4(i-1)+j}$, $1\leq j\leq 4$, denote respectively the four upper half-planes bounded by the above four lines (their index order is arbitrary). In this way, we have a set $H=\{h_1,h_2\ldots,h_{4m}\}$ of $4m$ upper half-planes.

For any $i$ and $j$ with $1\leq i\leq j\leq m$, we define $\alpha(i,j)$ as the $y$-coordinate of the lowest point in the common intersection of the upper half-planes of $H$ from $h_{4(i-1)+1}$ to $h_{4j}$, i.e., all upper half-planes defined by $u_t$ and $v_t$ for $t\in [i, j]$. Observe that if we use one center to cover all backbone and thorn vertices $u_t$ and $v_t$ for $t\in [i,j]$, then $\alpha(i,j)$ is equal to the optimal objective value of this one-center problem.

We define a matrix $M$ of dimension $m\times m$: For any $i$ and $j$ in $[1,m]$, if $i+j\leq m+1$, then $M[i,j]=\alpha[i,m+1-j]$; otherwise, $M[i,j]=0$.

For each twig $w_i$, we define two arrays $A^r_i$ and $A^l_i$ of at most $m$ elements each as follows. Let $h^+(w_i)$ and $h^-(w_i)$ denote respectively the upper half-planes bounded by the lines $l^+(w_i)$ and $l^-(w_i)$ defined in Section~\ref{sec:first}.
The array $A^r_i$ is defined on the vertices of $P$ on the right side of $v_i$, as follows. For each $j\in [1,m-i+1]$, if we use a single center to cover $w_i$ and all vertices $u_t$ and $v_t$ for $t\in [i,m+1-j]$,
then $A_i^r[j]$ is defined to be the optimal objective value of this one-center problem, which is equal to the $y$-coordinate of the lowest point in the common intersection of $h^+(w_i)$ and the upper half-planes of $H$ from $h_{4(i-1)+1}$ to $h_{4(m+1-j)}$. Symmetrically, array $A^l_i$ is defined on the left side of $v_i$. Specifically, for each $j\in [1,i]$, if we use one center to cover $w_i$ and all vertices $u_t$ and $v_t$ for $t\in [j,i]$,
then $A^l[j]$ is defined to be the optimal objective value, which is equal to the $y$-coordinate of the lowest point in the common intersection of $h^-(w_i)$ and the upper half-planes of $H$ from $h_{4(j-1)+1}$ to $h_{4i}$.

Let $\calM$ be the set of the matrices $M$ and $A_i^r$ and $A_i^l$ for all $1\leq i\leq m$. The following lemma implies that we can apply \msearch\ on $\calM$ to compute $\lambda^*$.

\begin{lemma}\label{lem:sorted}
Each matrix of $\calM$ is sorted, and $\lambda^*$ is an element of a matrix in $\calM$.
\end{lemma}
\begin{proof}
We first show that all matrices of $\calM$ are sorted.

Consider the matrix $M$. Consider two elements $M[i,j_1]$ and $M[i,j_2]$ in the same row with $j_1<j_2$. Our goal is to show that $M[i,j_1]\geq M[i,j_2]$.

\begin{enumerate}
\item
If $j_1>m+1-i$, then both $M[i,j_1]$ and $M[i,j_2]$ are zero. Thus, $M[i,j_1]\geq M[i,j_2]$ trivially holds.

\item
If $j_1\leq m+1-i< j_2$, then $M[i,j_2]=0$ and $M[i,j_1]=\alpha(i,m+1-j_1)$. By our way of defining upper half-planes of $H$, one can verify that $\alpha(i,m+1-j_1)\geq 0$. Therefore, $M[i,j_1]\geq M[i,j_2]$.

\item
If $j_2\leq m+1-i$, then $M[i,j_1]=\alpha(i,m+1-j_1)$ and $M[i,j_2]=\alpha(i,m+1-j_2)$. Let $H_1$ (resp., $H_2$) be the set of the upper half-planes of $H$ from $h_{4(i-1)+1}$ to $h_{4(m+1-j_1)}$ (resp.,  $h_{4(m+1-j_2)}$). Since $j_1<j_2$, $H_2$ is a subset of $H_1$, and thus the lowest point in the common intersection of the upper half-planes of $H_2$ is not higher than that of $H_1$. Hence, $\alpha(i,m+1-j_1)\geq \alpha(i,m+1-j_2)$ and thus   $M[i,j_1]\geq M[i,j_2]$.
\end{enumerate}

The above proves $M[i,j_1]\geq M[i,j_2]$. Therefore, all elements in each row are sorted in nonincreasing order.
By the similar approach we can show that all elements in each column are also sorted in nonincreasing order. We omit the details. Hence, $M$ is a sorted matrix.

Now consider an array $A_i^l$. Consider any two elements $A_i^l[j_1]$ and $A_i^l[j_2]$ with $j_1<j_2$. Our goal is to show that $A_i^l[j_1]\geq A_i^l[j_2]$.
The argument is similar as the above third case.
Let $H_1$ (resp., $H_2$) be the set of $h^-(w_i)$ and the upper half-planes of $H$ from $h_{4(j_1-1)+1}$ (resp., $h_{4(j_2-1)+1}$) to $h_{4i}$. Since $j_1<j_2$, $H_2$ is a subset of $H_1$ and the lowest point in the common intersection of the upper half-planes of $H_2$ is not higher than that of $H_1$. Hence,  $A_i^l[j_1]\geq A_i^l[j_2]$.

We can show that $A_i^r$ is also sorted in a similar way. We omit the details.

The above proves that every matrix of $\calM$ is sorted. In the following, we show that $\lambda^*$ must be an element of one of these matrices.

Imagine that we apply our feasibility test algorithm \ftest0\ on $\lambda=\lambda^*$ and the stem $P$ by considering $P$ as a tree with root $v_m$. Then, the algorithm will compute at most $k$ centers in $P$. The algorithm actually partitions $P$ into at most $k$ disjoint connected subtrees such that the vertices in each subtree is covered by the same center that is located in the subtree. Further, there must be a subtree $P_1$ that has a center $q$ and two vertices $v'$ and $v$ such that $w(v)\cdot d(v,q)=w(v')\cdot d(v',q)=\lambda^*$, since otherwise we could adjust the positions of the centers so that the maximum weighted distance from all vertices of $P$ to their closest centers would be strictly smaller than $\lambda^*$.
Since $P_1$ is connected and both $v$ and $v'$ are in $P_1$, the path $\pi(v,v')$ is also in $P_1$.
%Thus, all vertices of $P_1$ are covered by $q$.

Depending on whether one of $v$ and $v'$ is a twig vertex, there are two cases.

If neither vertex is a twig vertex, then we claim that all thorn vertices connecting to the backbone vertices of $\pi(v,v')$ are covered by the center $q$. Indeed, suppose $v_i$ is a backbone vertex in $\pi(v,v')$ and $v_i$ connects to a thorn vertex $u_i$. Assume to the contrary that $u_i$ is not covered by $q$. Recall that by the definition of thorns, $w(u_i)\cdot d(u_i,v_i)\leq \lambda_1$, and since $\lambda_1<\lambda^*$, we have $w(u_i)\cdot d(u_i,v_i)< \lambda^*$. According to \ftest0, $u_i$ is covered by a center $q'$ that is not on $e(u_i,v_i)$. Hence, $u_i$ and $q'$ is in a connected subtree, denoted by $P_2$, in the partition of $P$ induced by \ftest0. Clearly, $v_i$ is in $\pi(u_i,q')$. Since $P_2$ is connected and both $q'$ and $u_i$ are in $P_2$, every vertex of $\pi(u_i,q')$ is in $P_2$. Because $q'$ is not on $e(u_i,v_i)$, $v_i$ must be in $\pi(u_i,q')$ and thus is in $P_2$. However, since $v_i$ is in $\pi(v,v')$, $v_i$ is also in $P_1$. This incurs contradiction since $P_1\cap P_2=\emptyset$. This proves the claim.

If $v$ is a backbone vertex, then let $i$ be its index, i.e., $v=v_i$; otherwise, $v$ is a thorn vertex and let $i$ be the index such that $v$ connects the backbone vertex $v_i$. Similarly, define $j$ for $v'$. Without loss of generality, assume $i\leq j$.
The above claim implies that $\lambda^*$ is equal to the $y$-coordinate of the lowest point in the common intersection of all upper half-planes defined by the backbone vertices $v_t$ and thorn vertices $u_t$ for all $t\in [i,j]$, and thus, $\lambda^*=\alpha(i,j)$, which is equal to $M[i,m+1-j]$. Therefore, $\lambda^*$ is in the matrix $M$.

Next, we consider the case where at least one of $v$ and $v'$ is a twig vertex.
For each twig vertex $w_i$ of $P$, by definition, $w(w_i)\cdot d(w_i,v_i)\geq \lambda_2$, and since $\lambda^*\leq \lambda_2$, the twig $e(w_i,v_i)$ must contain a center. Because both $v$ and $v'$ are covered by $q$, only one of them is a twig vertex (since otherwise we would need two centers to cover them since each twig must contain a center). Without loss of generality, we assume that $v$ is a twig vertex, say, $w_i$. If $v'$ is a backbone vertex, then let $j$ be its index; otherwise, $v'$ is a thorn vertex and let $j$ be the index such that $v'$ connects the backbone vertex $v_j$. Without loss of generality, we assume that $i\leq j$.

By the same argument as the above, all thorn vertices $u_t$ with $t\in [i,j]$ are covered by $q$. This implies that $\lambda^*$ is the $y$-coordinate of the lowest point in the common intersection of $h^+(w_i)$ and all upper half-planes defined by the backbone vertices $v_t$ and thorn vertices $u_t$ for all $t\in [i,j]$. Thus,  $\lambda^*=A_i^r[m+1-j]$. Therefore, $\lambda^*$ is in the array $A_i^r$.

This proves that $\lambda^*$ must be in a matrix of $\calM$. The lemma thus follows.
\qed
\end{proof}

Note that $\calM$ consists of a matrix $M$ of dimension $m\times m$ and $2m$ arrays of lengths at most $m$. With the help of the 2D sublist LP query data structure in Lemma~\ref{lem:sublist}, the following lemma shows that the matrices of $\calM$ can be implicitly formed in $O(m\log m)$ time.

\begin{lemma}\label{lem:matrixform}
With $O(m\log m)$ time preprocessing, each matrix element of $\calM$ can be evaluated in $O(\log^2 m)$ time.
\end{lemma}
\begin{proof}
We build a 2D sublist LP query data structure of Lemma~\ref{lem:sublist} on the upper half-planes of $H$ in $O(m\log m)$ time. Then, each element of $M$ can be computed in $O(\log^2 m)$ time by a 2D sublist LP query.

Now consider an array $A_i^l$. Given any index $j$, to compute $A_i^l[j]$, recall that $A_i^l[j]$ is equal to the $y$-coordinate of the lowest point $p^*$ of the common intersection of the upper half-plane $h^+(w_i)$ and those in $H'$, where $H'$ is the set of the upper half-planes of $H$ from $h_{4(j-1)+1}$ to $h_{4i}$. The lowest point $p'$ of the common intersection of the upper half-planes of $H'$ can be computed in $O(\log^2 m)$ time by a 2D sublist LP query with query indices $4(j-1)+1$ and $4i$.
Computing $p^*$ can also be done in $O(\log^2 m)$ time by slightly modifying the query algorithm for computing $p'$. We briefly discuss it below and the interested reader should refer to~\cite{ref:ChenAp13}  for details (the proof of Lemma 8 and the discussion after the lemma).

The query algorithm for computing $p'$ is similar in spirit to the linear-time algorithm for the 2D linear programming problem in~\cite{ref:MegiddoLi83}. It is a binary search algorithm. In each iteration, the algorithm computes the highest intersection $p''$ between a vertical line $l$ and the bounding lines of the half-planes of $H'$, and based on the local information at the intersection, the algorithm will determine which side to proceed for the search. For computing $p^*$, we need to incorporate the additional half-plane $h^+(w_i)$. To this end, in each iteration of the binary search, after we compute the highest intersection $p''$, we compare it with the intersection of $l$ and the bounding line of $h^+(w_i)$ and update the highest intersection if needed. This costs only constant extra time for each iteration. Therefore, the total running time for computing $p^*$ is still $O(\log^2 m)$.

Computing the elements of arrays $A_i^r$ can be done similarly.
The lemma thus follows.
\qed
\end{proof}

By applying algorithm \msearch\ on $\calM$ with stopping count $c=0$ and $\kappa=O(\log^2 m)$, according to Lemma~\ref{lem:msearch}, \msearch\ produces $O(\log m)$ values for feasibility tests, and the total time exclusive of feasibility tests is $O(m\log^3 m)$ because we need to evaluate $O(m\log m)$ matrix elements of $\calM$. Hence, the total time for computing $\lambda^*$ is $O(m\log^3 m + \tau\cdot \log m)$.

\paragraph{Remark.}
Clearly, the first algorithm is better than the second one. However, later when
we use the techniques of the second algorithm, $m$ is often bounded by $O(\log^2 n)$ and thus $\log^3 m=O(\log n)$. In fact, we use the techniques of the second algorithm mainly because we need to set the stopping count $c$ to some non-zero value.

\section{Solving the $k$-Center Problem on $T$}
\label{sec:tree}

In this section, we present our algorithm for solving the $k$-center problem on $T$. We will focus on computing the optimal objective value $\lambda^*$.

Frederickson~\cite{ref:FredericksonOp91} proposed a {\em path-partition} of $T$, which is a partition of the edges of $T$ into paths where a vertex $v$ is an endpoint of a path if and only if the degree of $v$ in $T$ is not equal to $2$ (e.g., see Fig.~\ref{fig:pathpartition}). A path in a partition-partition of $T$ is called a {\em leaf-path} if it contains a leaf of $T$.

\begin{figure}[t]
\begin{minipage}[t]{\linewidth}
\begin{center}
\includegraphics[totalheight=1.3in]{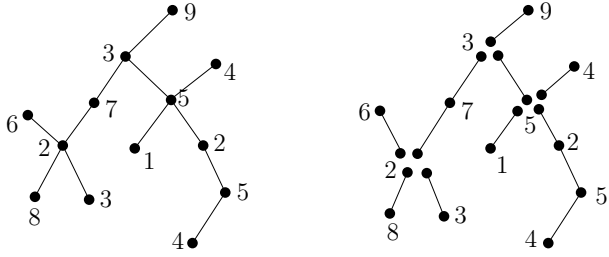}
\caption{\footnotesize Left: the tree $T$ where the numbers are the weights of the vertices. Right: the path partition of $T$.}
\label{fig:pathpartition}
\end{center}
\end{minipage}
\vspace*{-0.15in}
\end{figure}

As in~\cite{ref:FredericksonPa91}, we generalize the path-partition to stem-partition as follows. During the course of our algorithm, a range $(\lambda_1,\lambda_2]$ that contains $\lambda^*$ will be maintained and $T$ will be modified by removing some edges and adding some thorns and twigs. At any point in our algorithm, let $T'$ be $T$ with all thorns and twigs removed. A {\em stem} of $T$ is a path in the path-partition of $T'$, along with all thorns and twigs that connect to vertices in the path. A {\em stem-partition} of $T$ is to partition $T$ into stems according to a path-partition of $T'$. A stem in a stem-partition of $T$ is called a {\em leaf-stem} if it contains a leaf of $T$ that is a backbone vertex of the stem.

Our algorithm follows the first algorithmic scheme in~\cite{ref:FredericksonPa91}. There are two main phases: Phase 1 and Phase 2. Let $r=\log^2 n$. Phase 1 gathers information so that the feasibility test can be made in sublinear $O(\frac{n}{r}\log^3 r)$ time. Phase 2 computes $\lambda^*$ by using the faster feasibility test. If $T$ has more than $2n/r$ leaves, then there is an additional phase, called Phase 0, which reduces the problem to a tree with at most $2n/r$ leaves. (Phase 0 is part of Phase 1 in~\cite{ref:FredericksonPa91}, and we separates it from Phase 1 to make it clearer.) In the following, we consider the general case where $T$ has more than  $2n/r$ leaves. Algorithm~\ref{algo:kcenter} gives the pseudocode of the overall algorithm.

\subsection{The Preprocessing and Computing the Vertex Ranks}
\label{sec:preprocess}

We first perform some preprocessing. Recall that $\gamma$ is the root of $T$. We compute the distances $d(v,\gamma)$ for all vertices $v$ in $O(n)$ time. Then, if $u$ is an ancestor of $v$, $d(u,v)=d(\gamma,v)-d(\gamma,u)$, which can be computed in $O(1)$ time. In the following, whenever we need to compute a distance $d(u,v)$, it is always the case that one of $u$ and $v$ is an ancestor of the other, and thus $d(u,v)$ can be obtained in $O(1)$ time.

Next, we compute a ``rank'' $rank(v)$ for each vertex $v$ of $T$. These ranks will facilitate our algorithm later.
For each vertex $v$, we define a point $p(v)$ on the $x$-axis with $x$-coordinate equal to $d(\gamma,v)$ in an $xy$-coordinate system $\bbR^2$, and define $l(v)$ as the line through $p(v)$ with slope equal to $-w(v)$. Let $L$ be the set of these $n$ lines.
Consider the line arrangement $\calA(L)$ of $L$. Let $v_1(L)$ and $v_2(L)$ be the vertices as defined in Section~\ref{sec:pre}. By Lemma~\ref{lem:arrangement}, both vertices can be computed in $O(n\log n)$ time. Let $l$ be a horizontal line strictly between $v_1(L)$ and $v_2(L)$. We sort all lines of $L$ by their intersections with $l$ from left to right, and for each vertex $v$, we define $rank(v)=i$ if there are $i-1$ lines before $l(v)$ in the above order.
By the definitions of $v_1(L)$ and $v_2(L)$, the above order of $L$ is also an order of $L$ sorted by their intersections with the horizontal line $y=\lambda^*$.

\subsection{Phase 0}
\label{sec:phase0}

Recall that $T$ has more than $2n/r$ leaves. In this section, we reduce the problem to the problem of placing centers in a tree with at most $2n/r$ leaves. Our algorithm will maintain a range $(\lambda_1,\lambda_2]$ that contains $\lambda^*$.
Initially, $\lambda_1=y(v_2(L))$, the $y$-coordinate of $v_2(L)$, which is already computed in the preprocessing, and $\lambda_2=y(v_1(L))$. We form a stem-partition of $T$, which is actually a path-partition since there are no thorns and twigs initially, and this can be done in $O(n)$ time.

Recall that $r=\log^2 n$.
While there are more than $2n/r$ leaves in $T$, we do the following.

Recall that the length of a stem is defined as the number of backbone vertices.
Let $S$ be the set of all leaf-stems of $T$ whose lengths are at most $r$. Let $n'$ be the number of all backbone vertices on the leaf-stems of $S$. For each leaf-stem of $S$, we form matrices by Lemma~\ref{lem:matrixform}. Let $\calM$ denote the collection of matrices for all leaf-stems of $S$. We call \msearch\ on $\calM$, with stopping count $c=n'/(2r)$, by using the feasibility test algorithm \ftest0. After \msearch\ stops, we have an updated range $(\lambda_1,\lambda_2)$ and matrix elements of $\calM$ in $(\lambda_1,\lambda_2)$ are called {\em active} values. Since $c=n'/(2r)$, at most $n'/(2r)$ active values of $\calM$ remain, and thus at most $n'/(2r)$ leaf-stems of $S$ have active values.

For each leaf-stem $P\in S$ without active values, we perform the following {\em post-processing procedure}. The backbone vertex of $P$ closest to the root is called the {\em top vertex}. We place centers on $P$, subtract their number from $k$, and replace $P$ by either a thorn or a twig connected to the top vertex ($P$ is thus removed from $T$ except the top vertex), such that solving the $k$-center problem on the modified $T$ also solves the problem on the original $T$.
The post-processing procedure can be implemented in $O(m)$ time, where $m$ is the length of $P$. The details are given below.

\paragraph{The post-processing procedure on $P$.}
Let $z$ be the top vertex of $P$. We run the feasibility test algorithm \ftest0\ on $P$ with $z$ as the root and $\lambda=\lambda'$ that is an arbitrary value in $(\lambda_1,\lambda_2)$. After $z$ is finally processed, depending on whether $sup(z)\leq dem(z)$, we do the following.

If $sup(z)\leq dem(z)$, then let $q$ be the last center that has been placed.
In this case, all vertices of $P$ are covered and $z$ is covered by $q$.
According to algorithm \ftest0\ and as discussed in the proof of Lemma~\ref{lem:sorted}, $q$ covers a connected subtree of vertices, and let $V(q)$ denote the set of these vertices excluding $z$. Note that $V(q)$ can be easily identified during \ftest0. Let $k'$ be the number of centers excluding $q$ that have been placed on $P$. Since $\lambda'\in (\lambda_1,\lambda_2)$ and the matrices formed based on $P$ do not have any active values, we have the following {\em key observation:} if we run \ftest0\ with any $\lambda\in (\lambda_1,\lambda_2)$,
%then the algorithm will have the same behavior has that for $\lambda=(\lambda_1+\lambda_2)/2$, i.e., the algorithm with $\lambda=\lambda^*$
the algorithm will also cover all vertices of $P\setminus (V(q)\cup \{z\})$ with $k'$ centers and cover vertices of $V(q)\cup \{z\}$ with one center.
Indeed, this is true because the way we form matrices for $P$ is consistent with \ftest0, as discussed in the proof of Lemma~\ref{lem:sorted}. %Section~\ref{sec:second}.
In this case, we replace $P$ by attaching a twig $e(u,z)$ to $z$ with length equal to $d(u,z)$, where $u$ is a vertex of $V(q)$ with the following property: For any $\lambda\in (\lambda_1,\lambda_2)$, if we place a center $q'$ on the path $\pi(u,z)$ at distance $\lambda/w(u)$ from $u$, then $q'$ will cover all vertices of $V(q)$ under $\lambda$, i.e., $u$ ``dominates'' all other vertices of $V(q)$ and thus it is sufficient to keep $u$ (since $\lambda_2$ is feasible, any subsequent feasibility test in the algorithm will use $\lambda\in (\lambda_1,\lambda_2)$). The following lemma shows that $u$ is the vertex of $V(q)$ with the largest rank.

%Note that the center $q$ is either on the backbone or on a twig of $P$.
%
%If $q$ is on a twig $e(u,v)$, where $u$ is a twig vertex, by the definition of twigs and the above analysis, if we run \ftest0\ with $\lambda=\lambda^*$, the algorithm will place a center on $e(u,v)$ that can cover all vertices of $V(q)$.  Hence, we attach a twig $e(u,z)$ by connecting $u$ to $z$ with an edge of length equal to $d(u,z)$.
%%By the above discussion, the center placed on covering $u$ will cover all vertices of $V(q)$.
%
%If $q$ is on the backbone, then none of the vertex of $V(q)$ is a twig-vertex because a twig vertex can only be covered by a center on a twig. In this case, we need to select a vertex $v^*$ from $V(q)$ to form a twig with $z$ such that the center at the twig covering $v^*$ under distance $\lambda^*$ also covers all other vertices of $V(q)$. We pick the vertex $v\in V(q)$ with largest $rank(v)$ as $v^*$. The following lemma justifies our approach.

\begin{lemma}\label{lem:60}
Let $u$ be the vertex of $V(q)$ with the largest rank. For any $\lambda\in (\lambda_1,\lambda_2)$, the following holds.
\begin{enumerate}
\item
$\frac{\lambda}{w(u)}\leq d(u,z)$.
\item
If $q'$ is the point on the path $\pi(u,z)$ with distance $\frac{\lambda}{w(u)}$ from $u$, then $q'$ covers all vertices of $V(q)$ under $\lambda$, i.e., $w(v)\cdot d(v,q')\leq \lambda$ for all $v\in V(q)$.
\end{enumerate}
\end{lemma}
\begin{proof}
Before we prove $\frac{\lambda}{w(u)}\leq d(u,z)$, let $q'$ be the point on $\pi(u,\gamma)$ with distance $\frac{\lambda}{w(u)}$ from $u$ (if  $\frac{\lambda}{w(u)}>d(u,\gamma)$, then we add a dummy edge $e^*$ extended from the root $\gamma$ long enough so that $q'$ is on $e^*$). Later we will show that $\frac{\lambda}{w(u)}\leq d(u,z)$, which also proves that $q'$ is on $\pi(u,z)$.

We first show that $q'$ covers all vertices of $V(q)$.
Consider any vertex $v\in V(q)$. Our goal is to prove $w(v)\cdot d(v,q')\leq \lambda$. If $v=u$, this trivially holds. In the following, we assume $v\neq u$.

Note that $q'$ may be on a twig of $P$.
If $q'$ is on a twig $e(u,v)$ of $P$, then this means that $q'$ is on $P$ and $\frac{\lambda}{w(u)}\leq d(u,z)$ holds.
In this case, if we run \ftest0\ with $\lambda$, then $q'$ will be a center placed by \ftest0\ to cover $u$. On the other hand, according to the above key observation, \ftest0\ with $\lambda$ will use one center to cover all vertices of $V(q)$. Hence, $q'$ covers all vertices of $V(q)$ and thus covers $v$.
In the following, we assume that $q'$ is not on a twig of $P$. Note that by the definition of thorns, $q'$ cannot be in any thorn. Thus $q'$ must be either on the backbone of $P$ or outside $P$ in $\pi(z,\gamma)\cup e^*$. We define $V_1(q)$ to be the set of vertices of $V(q)$ in the subtree rooted at $q'$ and let $V_2(q)=V(q)\setminus V_1(q)$. Depending on whether $v$ is in $V_1(q)$ or $V_2(q)$, there are two cases.

\paragraph{The case $v\in V_1(q)$.}
Recall that in Section~\ref{sec:preprocess} each vertex $v'$ defines a line $l(v')$ in $\bbR^2$. We consider the two lines $l(v)$ and $l(u)$. Let $p_v$ and $p_u$ denote the intersections of the horizontal line $y=\lambda$ with $l(v)$ and $l(u)$, respectively (e.g., see Fig.~\ref{fig:rank}).
%For any point $p\in \bbR^2$, let $x(p)$ and $y(p)$ denote the $x$- and $y$-coordinates of $p$, respectively.
Note that the point $q'$ corresponding to the point $(x(p_u),0)$ in the $x$-axis in the sense that $d(\gamma,q')=x(p_u)$. Since $rank(v)<rank(u)$ and $\lambda\in (\lambda_1,\lambda_2)$, by the definition of ranks, it holds that $x(p_v)< x(p_u)$. Because the slope of $l(v)$ is not positive, $y(p_v')\leq \lambda$, where $p_v'$ is the intersection of $l(v)$ with the vertical line through $p_u$. On the other hand, since $v$ is in $V_1(q)$, $y(p_v')$ is exactly equal to $w(v)\cdot d(v,q')$. Therefore, we obtain $w(v)\cdot d(v,q')\leq \lambda$.

\begin{figure}[t]
\begin{minipage}[t]{\linewidth}
\begin{center}
\includegraphics[totalheight=1.3in]{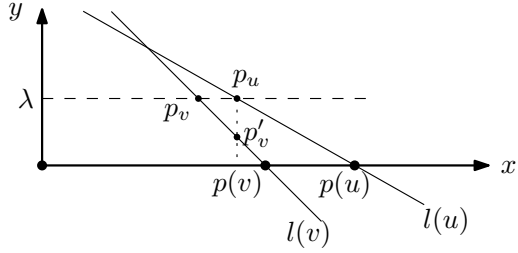}
\caption{\footnotesize Illustrating the proof of the case $v\in V_1(q)$. Note that $p(v)$ and $p(u)$ are the points defined respectively by $v$ and $u$ in Section~\ref{sec:preprocess}.}
\label{fig:rank}
\end{center}
\end{minipage}
\vspace*{-0.15in}
\end{figure}

\paragraph{The case $v\in V_2(q)$.}
In this case, $V_2(q)\neq \emptyset$. By the definition of $V_2(q)$, $q'$ must be in $\pi(u,v)$. According to the above key observation, we can use one center to cover all vertices of $V(q)$ (under $\lambda$), and in particular, we can use one center to cover both $u$ and $v$. By the definition of $q'$, $q'$ is the closest point to $v$ on $\pi(u,v)$ that can cover $u$. Hence, $q'$ must be able to cover $v$. Therefore, we obtain $w(v)\cdot d(v,q')\leq \lambda$.
\vspace{0.1in}

The above proves $w(v)\cdot d(v,q')\leq \lambda$.

Finally, we argue that $\frac{\lambda}{w(u)}\leq d(u,z)$. Assume to the contrary that this is not true. Then, $q'$ is outside $P$. This means that we can place a center outside $P$ to cover all vertices of $V(q)$ under $\lambda$. But this contradicts the above key observation that \ftest0\ for $\lambda'$ will place a center in $P$ to cover the vertices of $V(q)$. The lemma thus follows.
\qed
\end{proof}

Due to the preprocessing in Section~\ref{sec:preprocess}, we can find $u$ from $V(q)$ in $O(m)$ time. This finishes our post-processing procedure for the case $sup(z)\leq dem(z)$.
%One can verify that due to the definition of $u$, $w(u)\cdot d(u,z)$ is
Since $\frac{\lambda}{w(u)}\leq d(u,z)$ for any $\lambda\in (\lambda_1,\lambda_2)$,  we have $w(u)\cdot d(u,z)\geq \lambda_2$, and thus, $e(u,z)$ is indeed a twig.

Next, we consider the other case $sup(z)>dem(z)$. In this case, $P$ has some vertices other than $z$ that are not covered yet, and we would need to place a center at $z$ to cover them. Let $V$ be the set of all uncovered vertices other than $z$, and $V$ can be identified during \ftest0. In this case, we replace $P$ by attaching a thorn $e(u,z)$ to $z$  with length equal to $d(u,z)$, where $u$ is a vertex of $V$ with the following property: For any $\lambda\in (\lambda_1,\lambda_2)$, if there is a center $q$ outside $P$ covering $u$ {\em through} $z$ (by ``through'', we mean that $\pi(q,u)$ contains $z$) under distance $\lambda$, then $q$ also covers all other vertices of $V$ (intuitively $u$ ``dominates'' all other vertices of $V$). Since later we will place centers outside $P$ to cover the vertices of $V$ through $z$ under some $\lambda\in (\lambda_1,\lambda_2)$, it is sufficient to maintain $u$. The following lemma shows that $u$ is the vertex of $V$ with the largest rank.

\begin{lemma}\label{lem:70}
Let $u$ be the vertex of $V$ with the largest rank. Then, for any center $q$ outside $P$ that covers $u$ through $z$ under any distance $\lambda\in (\lambda_1,\lambda_2)$, $q$ also covers all other vertices of $V$.
\end{lemma}
\begin{proof}
Let $v$ be any vertex of $V$ other than $u$. Our goal is to prove that $d(q,v)\cdot w(v)\leq \lambda$. The proof is similar to that for the case $v\in V_1(q)$ of Lemma~\ref{lem:60} and we omit the details.
\qed
\end{proof}

Since a center at $z$ would cover $u$, it holds that $w(u)\cdot d(u,z)\leq \lambda$ for any $\lambda\in (\lambda_1,\lambda_2)$, which implies that $w(u)\cdot d(u,z)\leq \lambda_1$. Thus, $e(u,z)$ is indeed a thorn.

The above replaces $P$ by attaching to $z$ either a thorn or a twig. We perform the following additional processing.

Suppose $z$ is attached by a thorn $e(z,u)$. If $z$ already has another thorn $e(z,u')$, then we discard one of $u'$ and $u$ whose rank is smaller, because any center that covers the remaining vertex will cover the discarded one as well (the proof is similar to those in Lemma~\ref{lem:60} and \ref{lem:70} and we omit it). This makes sure that $z$ has at most one thorn.

Suppose $z$ is attached by a twig $e(z,u)$. If $z$ already has another twig $e(z,u')$, then we can discard one of $u$ and $u'$ whose rank is {\em larger} (and subtract $1$ from $k$). The reason is the following. Without loss of generality, assume $rank(u)<rank(u')$. Since both $e(z,u)$ and $e(z,u')$ are twigs, if we apply \ftest0\ on any $\lambda\in (\lambda_1,\lambda_2)$, then the algorithm will place a center $q$ on $e(z,u)$ with distance $\lambda/w(u)$ from $u$ and place a center $q'$ on $e(z,u')$ with distance $\lambda/w(u')$ from $u'$.
As $rank(u)<rank(u')$, we have the following lemma.

\begin{lemma}\label{lem:80}
$d(q,z)\leq d(q',z)$.
\end{lemma}
\begin{proof}
The analysis is similar to those in Lemma~\ref{lem:60} and \ref{lem:70}.
Consider the lines $l(u)$ and $l(u')$ in $\bbR^2$ defined by $u$ and $u'$, respectively, as discussed in Section~\ref{sec:preprocess}. Let $p_u$ and $p_{u'}$ be the intersections of the horizontal line $y=\lambda$ with $l(u)$ and $l(u')$, respectively (e.g., see Fig.~\ref{fig:rank1}). Since $rank(u)<rank(u')$ and $\lambda\in (\lambda_1,\lambda_2)$, $x(p_u)\leq x(p_{u'})$. Note that
%$q$ corresponds to the point $(x(p_u),0)$ in the sense that
$x(p_u)=d(\gamma,q)$ and %similarly, $q'$ corresponds to the point $(x(p_{u'}),0)$ with
$x(p_{u'})=d(\gamma,q')$. Since $x(p_u)\leq x(p_{u'})$, we have $d(\gamma,q)\leq d(\gamma,q')$.

\begin{figure}[t]
\begin{minipage}[t]{\linewidth}
\begin{center}
\includegraphics[totalheight=1.3in]{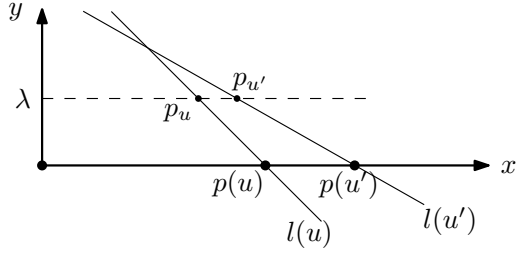}
\caption{\footnotesize Illustrating the proof of Lemma~\ref{lem:80}. Note that $p(u)$ and $p(u')$ are the points defined respectively by $u$ and $u'$ in Section~\ref{sec:preprocess}.}
\label{fig:rank1}
\end{center}
\end{minipage}
\vspace*{-0.15in}
\end{figure}

On the other hand, due to that $q\in e(z,u)$ and $q'\in e(z,u')$, $d(\gamma,z)\leq d(\gamma,q)$ and $d(\gamma,z)\leq d(\gamma,q')$. Thus, $d(q,z)=d(\gamma,q)-d(\gamma,z)$ and $d(q',z)=d(\gamma,q')-d(\gamma,z)$. Because $d(\gamma,q)\leq d(\gamma,q')$, we obtain $d(q,z)\leq d(q',z)$. The lemma thus follows.
\qed
\end{proof}

Lemma~\ref{lem:80} tells that any vertex that is covered by $q'$ in the subsequent algorithm will also be covered by $q$. Thus, it is sufficient to maintain the twig $e(z,u)$. Since we need to place a center at $e(z,u')$, we subtract $1$ from $k$ after removing $e(z,u')$.  Hence, $z$ has at most one twig.

This finishes the post-processing procedure for $P$. Due to the preprocessing in Section~\ref{sec:preprocess}, the running time of the procedure is $O(m)$.

\vspace{0.15in}

Let $T$ be the modified tree after the post-processing on each stem $P$ without active values. If $T$ still has more than $2n/r$ leaves, then we repeat the above. The algorithm stops once $T$ has at most $2n/r$ leaves. This finishes Phase 0. The following lemma gives the time analysis, excluding the preprocessing in Section~\ref{sec:preprocess}.
%We define the {\em length} of a stem $P$ as the number of its backbone vertices.

\begin{lemma}\label{lem:90}
Phase 0 runs in $O(n(\log\log n)^3)$ time.
\end{lemma}
\begin{proof}
We first argue that the number of iterations of the while loop is $O(\log r)$. The analysis is very similar to those in~\cite{ref:FredericksonOp91,ref:FredericksonPa91}, and we include it here for completeness.

%We define the {\em length} of a stem $P$ as the number of backbone vertices in $P$.
We consider an iteration of the while loop.
Suppose the number of leaf-stems in $T$, denoted by $m$, is at least $2n/r$.
Then, at most $n/r$ leaf-stems are of length larger than $r$. Hence, at least half of the leaf-stems are of length at most $r$. Thus, $|S|\geq m/2$. Recall that $n'$ is the total number of backbone vertices in all leaf-stems of $S$. Because at most $n'/(2r)$ leaf-stems have active values after \msearch, at least $|S|-n'/(2r)\geq m/2-n'/(2r)\geq m/2-n/(2r)\geq m/2-m/4=m/4$ leaf-stems will be removed. Note that removing two such leaf-stems may make an interior vertex become a new leaf in the modified tree. Hence, the tree resulting at the end of each iteration will have at most $7/8$ of the leaf-stems of the tree at the beginning of the iteration. Therefore, the number of iterations of the while loop needed to reduce the number of leaf-stems to at most $2n/r$ is $O(\log r)$.

We proceed to analyze the running time of Phase 0.
In each iteration of the while loop, we call \msearch\ on the matrices
for all leaf-stems of $S$. Since the length of each stem $P$ of $S$ is
at most $r$, there are $O(r)$ matrices formed for $P$. We perform the
preprocessing of Lemma~\ref{lem:matrixform} on the matrices, so that
each matrix element can be evaluated in $O(\log^2 r)$ time. The total
time of the preprocessing on stems of $S$ is $O(n'\log r)$.
Since $\calM$ has $O(n')$ matrices and the stopping account $c$ is
$n'/(2r)$, each call to \msearch\ produces $O(\log r)$ values
for feasibility tests in $O(n'\log^3 r)$ time (i.e., $O(n'\log r)$
matrix elements will be evaluated).
For each leaf-stem without active values, the post-processing time for it is $O(r)$. Hence, the total post-processing time in each iteration is $O(n')$.

Since there are $O(\log r)$ iterations, the total number of
feasibility tests is $O(\log^2 r)$, and thus the overall time for all
feasibility tests in Phase 0 is $O(n\log^2 r)$. On the other hand,
after each iteration, at most $n'/(2r)$ leaf-stems of $S$ have active
values and other leaf-stems of $S$ will be deleted. Since the length
of each leaf-stem of $S$ is at most $r$, the leaf-stems with active
values have at most $n'/2$ backbone vertices, and thus at least $n'/2$
backbone vertices will be deleted in each iteration. Therefore, the
total sum of such $n'$ in all iterations is $O(n)$. Hence, the total
time for the preprocessing of Lemma~\ref{lem:matrixform} is $O(n\log
r)$, the total time for \msearch\ is
$O(n\log^3 r)$, and the total post-processing time for leaf-stems
without active values is $O(n)$.

In summary, the overall time of Phase 0 (excluding the preprocessing in Section~\ref{sec:preprocess}) is $O(n\log^3 r)$, which is $O(n(\log\log n)^3)$ since $r=\log^2 n$.
\qed
\end{proof}

\subsection{Phase 1}

We assume that the tree $T$ now has at most $2n/r$ leaves and we want to place $k$ centers in $T$ to cover all vertices. Note that $T$ may have some thorns and twigs.
The main purpose of this phase is to gather information so that each feasibility test can be done in sublinear time, and specifically, $O(n/r\log^3 r)$ time. Recall that we have a range $(\lambda_1,\lambda_2]$ that contains $\lambda^*$.

We first form a stem-partition for $T$.
%For each stem except the one containing the root $r$, we remove the top vertex from it.
Then, we further partition the stems into substems, each of length at most $r$, such that the lowest backbone vertex $v$ in a substem is the highest backbone vertex in the next lower substem (if $v$ has a thorn or/and a twig, then they are included in the upper substem). So this results in a partition of edges. Let $S$ be the set of all substems. Let $T_c$ be the tree in which each node represents a substem of $S$ and node $\mu$ in $T_c$ is the parent of node $\nu$ if the highest backbone vertex of the substem for $\nu$ is the lowest backbone vertex of the substem for $\mu$, and we call $T_c$ the {\em stem tree}.
As in~\cite{ref:FredericksonOp91,ref:FredericksonPa91}, since $T$ has at most $2n/r$ leaves, $|S|=O(n/r)$ and the number of nodes of $T_c$ is $O(n/r)$.

For each substem $P\in S$, we compute the set $L_P$ of lines as in Section~\ref{sec:first}. Let $L$ be the set of all the lines for all substems of $S$. We define the lines of $L$ in the same $xy$-coordinate system $\bbR^2$. Clearly, $|L|=O(n)$. Consider the line arrangement $\calA(L)$. Define vertices $v_1(L)$ and $v_2(L)$ of $\calA(L)$ as in Section~\ref{sec:pre}. With Lemma~\ref{lem:arrangement} and \ftest0, both vertices can be computed in $O(n\log n)$ time. We update $\lambda_1=\max\{\lambda_1,y(v_2(L))\}$ and $\lambda_2=\min\{\lambda_2,y(v_1(L))\}$. Hence, we still have $\lambda^*\in (\lambda_1,\lambda_2]$. We again call the values in $(\lambda_1,\lambda_2)$ {\em active} values.

For each substem $P\in S$, observe that each element of the matrices formed based on $P$ in Section~\ref{sec:second} is equal to the $y$-coordinate of the intersection of two lines of $L_P$, and thus is equal to the $y$-coordinate of a vertex of $\calA(L)$. By the definitions of $v_1(L)$ and $v_2(L)$, no matrix element of $P$ is active.
%and all possible candidate values produced by $P$ are not active.
%
%Next, we form the matrices for all stems of $S$, and let $\calM$ denote the collection of all matrices. We call algorithm \msearch\ on $\calM$ with stopping count $c=0$, with \ftest0\ as the feasibility test algorithm. After \msearch, we obtain a range $(\lambda_1,\lambda_2]$ that contains $\lambda^*$. Because $c=0$, no substem of $S$ has a active value.
%We form a {\em stem-tree} $T_c$, in which each node represents a substem of $S$ and a node $\mu$ is a parent of a node $\nu$ if the top backbone vertex in the substem for $\nu$ is the bottom backbone vertex in the substem for $\mu$.
%
%As analyzed in~\cite{ref:FredericksonPa91}, since $T$ has at most $2n/r$ leaves, the stem-partition of $T$ has $O(n/r)$ sub-stems, i.e., $|S| = O(n/r)$. Forming the matrices for all substems of $S$ takes $O(n\log r)$ time by Lemma~\ref{lem:matrixform}. Algorithm \msearch\ produces $O(\log n)$ search values in $O(n\log^3 r)$ time. Using \ftest0, all feasibility tests takes $O(n\log n)$ time. Forming the tree $T_c$ can be easily done in $O(n)$ time. Hence, the total time for all above work is $O(n\log n)$.
%

In the future algorithm, we will only need to test feasibilities for values $\lambda\in (\lambda_1,\lambda_2)$.
In what follows, we compute a data structure on each substem $P$ of $S$, so that it will help make the feasibility test faster. We will prove the following lemma and use \ftestnew\ to denote the feasibility test algorithm in the lemma.

\begin{lemma}
After $O(n\log n)$ time preprocessing, each feasibility test can be
done in $O(n/r\cdot \log^3 r)$ time.
\end{lemma}

We first discuss the preprocessing and then present the algorithm \ftestnew.

\subsubsection{The Preprocessing for \ftestnew}

Consider a substem $P$ of $S$. Let $v_1,v_2,\ldots,v_m$ be the backbone vertices of $P$ sorted from left to right, with $v_m$ as the top vertex. Each vertex $v_i$ may have a twig $e(w_i,v_i)$ and a thorn $e(u_i,v_i)$.
Let $\lambda$ be an arbitrary value in $(\lambda_1,\lambda_2)$. In the
following, all statements made to $\lambda$ is applicable to
any $\lambda\in (\lambda_1,\lambda_2)$, and this is due to that
none of the elements in the matrices produced by $P$ is active.

By the definition of twigs, if we run \ftest0\ with $\lambda$, the
algorithm will place a center, denoted by $q_i$, on each
twig $e(w_i,v_i)$ at distance $\lambda/w(w_i)$ from $w_i$
%(note that
%the post-processing procedure guarantees that $d(w_i,v_i)\leq
%\lambda/w(w_i)$ for any $\lambda\in (\lambda_1,\lambda_2)$).
We first run the following {\em cleanup procedure} to remove all vertices of $P$
that can be covered by the centers on the twigs under $\lambda$.

\paragraph{The cleanup procedure.}
We first compute a rank $rank'(l)$ for each line $l$ of $L$, as follows.

%Let $l^*$ be the horizontal line whose $y$-coordinate is $\lambda$.
Let $L'$ be the sequence of the lines of $L$ sorted by their intersections with the horizontal line $y=\lambda$ from left to right. By the definitions of $\lambda_1$ and $\lambda_2$, $L'$ is also the sequence of the lines of $L$ sorted by their intersections with the horizontal line $y=\lambda^*$. In fact, the sequence $L'$ is unique for any $\lambda\in (\lambda_1,\lambda_2)$. For any line $l\in L$, if there are $i-1$ lines before $l$ in $L'$, then we define $rank'(l)$ to be $i$. Clearly, $rank'(l)$ for all lines $l\in L$ can be computed in $O(n)$ time.

Consider a twig $e(w_i,v_i)$ and a backbone vertex $v_j$ with $j\geq i$. Recall that $w_i$ defines a line $l^+(w_i)$ of slope $w(w_i)$ and $v_j$ defines a line $l^-(v_j)$ of slope $-w(v_j)$ in $L_P$ (and thus are in $L$). We have the following lemma.

\begin{figure}[t]
\begin{minipage}[t]{\linewidth}
\begin{center}
\includegraphics[totalheight=1.3in]{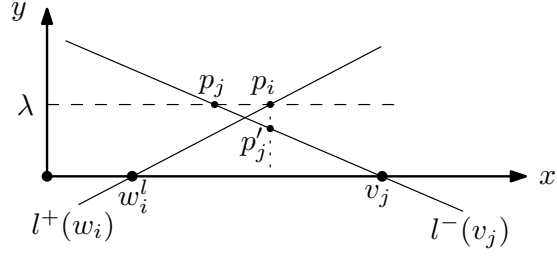}
\caption{\footnotesize Illustrating the proof of Lemma~\ref{lem:110}. Note that $w_i^l$ is the point defined by $w_i$, as discussed in Section~\ref{sec:first}.}
\label{fig:rank2}
\end{center}
\end{minipage}
\vspace*{-0.15in}
\end{figure}

\begin{lemma}\label{lem:110}
For any $\lambda\in (\lambda_1,\lambda_2)$, the center $q_i$ on $e(w_i,v_i)$  covers $v_j$ if and only if $rank'(l^+(w_i))>rank'(l^-(v_j))$.
\end{lemma}
\begin{proof}
Let $p_i$ and $p_j$ be the intersections of the horizontal line $y=\lambda$ with $l^+(w_i)$ and $l^-(v_j)$, respectively. Refer to Fig.~\ref{fig:rank2}. Let $p_j'$ denote the intersection of $l^-(v_j)$ with the vertical line through $p_i$. Since $q_i$ is located on $e(w_i,v_i)$, according to the definitions of $l^+(w_i)$ and $l^-(v_j)$, $y(p_j')$ is exactly equal to $w(v_j)\cdot d(q_i,v_j)$. Hence, $q_i$ covers $v_j$ if and only if $p_j'$ is below the line $y=\lambda$. On the other hand, $p_j'$ is below the line $y=\lambda$ if and only if $p_i$ is to the right of $p_j$, i.e., $rank'(l^+(w_i))>rank'(l^-(v_j))$. The lemma thus follows. \qed
\end{proof}

Consider a thorn vertex $u_j$ with $j\geq i$. Recall that $u_j$ defines a line $l^-(u_j)$ in $L_P$ with slope $-w(u_j)$. Similarly as above, $q_i$ covers $u_j$ if and only if $rank'(l^+(w_i))>rank'(l^-(u_j))$.

Based on the above observations, we use the following algorithm to find all backbone and thorn vertices of $P$ that can be covered by the centers on the twigs to their left sides. Let $i'$ be the smallest index such that $v_{i'}$ has a twig $w_{i'}$. The algorithm maintains an index $t$. Initially, $t=i'$. For each $i$ incrementally from $i'$ to $m$, we do the following. If $v_i$ has a twig-vertex $w_i$, then we reset $t$ to $i$ if $rank'(l^+(w_i))>rank'(l^+(w_t))$. The reason we do so is that if $rank'(l^+(w_i))>rank'(l^+(w_t))$, then for any $j\geq i$ such that $v_j$ (resp., $u_j$) is covered by the center $q_t$ on the twig $e(w_t,v_t)$, $v_j$ (resp., $u_j$) is also covered by the center $q_i$ on the twig $e(w_i,v_i)$, and thus it is sufficient to maintain the twig $e(w_i,v_i)$. Next, if $rank'(l^+(w_t))>rank'(l^-(v_i))$, then we mark $v_i$ as ``covered''. If $u_i$ exits and $rank'(l^+(w_t))>rank'(l^-(u_i))$, then we mark $u_i$ as ``covered''.

The above algorithm runs in $O(m)$ time and marks all vertices $v_j$ (reps., $u_j$) such that there exists a twig $e(w_i,v_i)$ with $i\leq j$ whose center $q_i$ covers $v_j$ (reps., $u_j$). In a symmetric way by scanning the vertices from $v_m$ to $v_1$, we can mark in $O(m)$ time all vertices $v_j$ (reps., $u_j$) such that there exits a twig $e(w_i,v_i)$ with $i\geq j$ whose center $q_i$ covers $v_j$ (reps., $u_j$). We omit the details. This marks all vertices that are covered by centers on twigs.

Let $V$ be the set of backbone and thorn vertices of $P$ that are not
marked, which are vertices of $P$ that need to be covered by placing
centers on the backbone of $P$ or outside $P$. If a thorn vertex $u_i$
is in $V$ but its connected backbone vertex $v_i$ is not in $V$, this
means that $v_i$ is covered by a center on a twig while $u_i$ is not
covered by any such center. Observe that any center on the backbone of $P$ or outside $P$ that covers $u_i$ will cover $v_i$ as well. For convenience of discussion, we include such $v_i$ into $V$ as well. Let $v_1',v_2',\ldots,v_t'$ be the backbone vertices of $V$ sorted from left to right (i.e., $v_t'$ is closer to the root of $T$). Note that $v_1'$ may not be $v_1$ and $v_t'$ may not be $v_m$. If $v'_i$ has a thorn vertex in $V$, then we use $u_i'$ to denote it.

This finishes the cleanup procedure.
\vspace{0.15in}

Next, we compute a data structure for $P$
to maintain some information for faster feasibility tests.

First of all, we maintain the index $a$ of the twig vertex $w_{a}$
such that $rank'(l^-(w_a))<rank'(l^-(w_i))$ for any other twig vertex
$w_i$ of $P$. The reason we keep $a$ is the following. Observe that for any
vertex $v$ that is a descent vertex of $v_1$ in $T$ (so $v$ is in
another substem that is a descent substem of $P$ in the stem tree
$T_c$), if $v$ can be covered by the center $q_i$ on a twig
$e(v_i,w_i)$ under any $\lambda\in (\lambda_1,\lambda_2)$, then $v$
can also be covered by the center $q_a$ on the twig $e(v_a,w_a)$ under
$\lambda$.
Symmetrically, we maintain the index $b$ of the twig vertex $w_{b}$
such that $rank'(l^+(w_b))>rank'(l^+(w_i))$ for any other twig vertex
$w_i$. Similarly, this is because for any vertex $v$ that is not in
any substem of the subtree rooted at $P$ in $T_c$, if $v$ can be
covered by the center on the twig $e(v_i,w_i)$ under $\lambda$, then
$v$ can also be covered by the center on the twig $e(v_b,w_b)$ under
$\lambda$.
Both $a$ and $b$ can be computed in $O(m)$ time.

For any two indices $i$ and $j$ with $1\leq i\leq j\leq t$, we use $V[i,j]$ to denote the set of all backbone vertices $v'_l$ and thorn vertices $u'_l$ with $l\in [i,j]$.

For each index $i\in [1,t]$, we maintain an integer $ncen(i)$ and a vertex $v(i)$ of $V$, which we define below.
Roughly speaking, $ncen(i)$ is the minimum number of centers that are
needed to cover all vertices of $V[i,t]$, minus one, and if we use
$ncen(i)$ centers to cover as many vertices of $V[i,t]$ as possible
from left to right,
$v(i)$ is the vertex that is not covered but ``dominates'' all other
uncovered vertices under $\lambda$.
Their detailed definitions are given below.

Let $j_i$ be the smallest index in $[i,t]$ such that it is not possible to cover all vertices in $V[i,j_i]$ by one center under $\lambda$.
If such a index $j_i$ does not exist, we let $j_i=t+1$.

If $j_i=t+1$, then we define $ncen(i)=0$. We define $v(i)$ as the vertex $v$ in $V[i,t]$ such that $rank'(l^+(v))<rank'(l^+(v'))$ for any other vertex $v'\in V[i,t]$. The reason we maintain such $v(i)$ is as follows.
Suppose during a feasibility test with $\lambda$, all vertices of $V[1,i-1]$ have been covered and we need to place a new center to cover those in $V[i,t]$. According to the greedy strategy of \ftest0, we want to place a center as close to the root $\gamma$ as possible. There are two cases.

In the first case, $w(v(i))\cdot d(v(i),v_m)< \lambda$, and one can verify that
if we place a center at the top vertex $v_m$, it can cover all vertices of $V[i,t]$ under $\lambda$. In this case, we do not place a center on the backbone of $P$ but will use a center outside $P$ to cover them (more precisely, this center is outside the subtree of $T_c$ rooted at substem $P$). We maintain $v(i)$ because any center outside $P$ covering $v(i)$ will cover all other vertices of $V[i,t]$ as well.

In the second case, $w(v(i))\cdot d(v(i),v_m)\geq \lambda$, and we need to place a center on the backbone of $P$. Again, according to the greedy strategy of \ftest0, we want to place this center close to $v_m$ as much as possible, and we use $q^*$ to denote such a center. The following lemma shows that $q^*$ is determined by $v(i)$.

\begin{lemma}\label{lem:120}
$q^*$ is on the path $\pi(v(i),v_m)$ of distance $\frac{\lambda}{w(v(i))}$ from $v(i)$.
\end{lemma}
\begin{proof}
The proof is somewhat similar to that of Lemma~\ref{lem:60}.

Let $p$ be the point on $\pi(v(i),v_m)$ of distance
$\frac{\lambda}{w(v(i))}$ from $v(i)$. Note that since $w(v(i))\cdot
d(v(i),v_m)\geq \lambda$, such a point $p$ must exist on
$\pi(v(i),v_m)$. By definition, $p$ is the point on $\pi(v(i),v_m)$
closest to $v_m$ that can cover $v(i)$. If $v(i)$ is a backbone
vertex, then $p$ is on the backbone of $P$. Otherwise, $v(i)$ is a
thorn vertex and $w(v(i))\cdot d(v(i), v)\leq \lambda_1$, where $v$ is the backbone
vertex that connects $v(i)$. Since $\lambda_1<\lambda$, we obtain
$w(v(i))\cdot d(v(i), v)< \lambda$, and thus $p$ must
be on the backbone of $P$. Hence, in either case, $p$ is on the
backbone of $P$. Consider any vertex $v'\in V[i,t]$ with $v'\neq
v(i)$. In the following, we show that $v'$ is covered by $p$.

If $v'$ is in the subtree rooted at $p$ (i.e., $\pi(v',v_m)$ contains
$p$), then since $rank'(l^+(v(i)))<rank'(l^+(v'))$, one can verify
that $v'$ is covered by $p$. Otherwise, assume to the contrary that
$p$ does not cover $v'$. Then, we would need to move $p$ towards $v'$ in order to cover $v'$. However,
since $p$ is the point on $\pi(v(i),v_m)$ closest to $v_m$ that can
cover $v(i)$, $p$ is also the point on $\pi(v(i),v')$ closest to $v'$
that can cover $v(i)$. Hence, moving $p$ towards $v'$ will make $p$
not cover $v(i)$ any more, which implies that no point on the backbone
of $P$ can cover both $v'$ and $v(i)$. This contradicts with the fact
that it is possible to place a center on the backbone of $P$ to cover
all vertices in $V[i,t]$. Therefore, $p$ covers $v'$.

The above shows that $p$ is the  point  on the backbone of $P$ closest
to $v_m$ that can cover all vertices of $V[i,t]$. Thus, $p$ is $q^*$,
and the lemma follows.
\qed
\end{proof}

The above defines $ncen(i)$ and $v(i)$ for the case where $j_i=t+1$. If $j_i\leq t$, we define $ncen(i)$ and $v(i)$ recursively as $ncen(i)=ncen(j_i)+1$ and $v(i)=v(j_i)$. Note that $i<j_i$, and thus this recursive definition is valid.

In the following, we present an algorithm to compute $ncen(i)$ and $v(i)$  for all $i\in [1,t]$. In fact, the above recursive definition implies a dynamic programming approach to scan the vertices $v'_i$ backward from $t$ to $1$. The details are given in the following lemma.

\begin{lemma}\label{lem:130}
$ncen(i)$ and $v(i)$ for all $i\in [1,t]$ can be computed in $O(m\log^2 m)$ time.
\end{lemma}
\begin{proof}
For any $i$ and $j$ with $1\leq i\leq j\leq t$, consider the following one-center problem: find a center to cover all backbone and thorn vertices of $V[i,j]$. As discussed in Section~\ref{sec:pre}, each backbone or thorn vertex defines two upper half-planes such that the optimal objective value for the above one-center problem is equal to the $y$-coordinate of the lowest point in the common intersection of the at most $4(j-i+1)$ half-planes defined by the backbone and thorn vertices of $V[i,j]$.
As in Section~\ref{sec:pre}, as preprocessing, we first compute the upper half-planes defined by all vertices of $V$ and order them by the indices of their corresponding vertices in $V$, and then compute the 2D sublist LP query data structure of Lemma~\ref{lem:sublist} in $O(t\log t)$ time. As in Section~\ref{sec:pre}, the lowest point of the common intersection of the upper half-planes defined by vertices of $V[i,j]$ can be computed by a 2D sublist LP query in $O(\log^2 t)$ time. We use $\alpha(i,j)$ to denote the optimal objective value of the above one-center problem for $V[i,j]$. With the above preprocessing, given $i$ and $j$, $\alpha(i,j)$ can be computed in $O(\log^2 t)$ time.

We proceed to compute $ncen(i)$ and $v(i)$ for all $i\in [1,t]$.

For each $i$ from $t$ downto $1$, we do the following. We maintain the index $j_i$. Initially when $i=t$, we set $j_i=t+1$, $ncen(i)=0$, and $v(i)=v'_t$. %When $i_j=t+1$, we also maintain a vertex $v$ which is $v(i)$ defined before for the case where $i_j$ does not exist. Initially, $v=v'_t$.
We process index $i$ as follows. We first compute $\alpha(i,j_i-1)$ in $O(\log^2 t)$ time. Depending on whether $\alpha(i,j_i-1)\leq \lambda$, there are two cases.

If $\alpha(i,j_i-1)\leq \lambda$, then depending on whether $j_i\neq t+1$, there are two subcases.

If $j_i\neq t+1$, then $ncen(i)=ncen(j_i)+1$ and $v(i)=v(j_i)$. Otherwise, we first set $ncen(i)=0$ and $v(i)=v(i+1)$. If $rank'(l^+(v'_i))<rank'(l^+(v(i)))$, then we reset $v(i)$ to $v'_i$. Further, if $v'_i$ has a thorn $u'_i$ and $rank'(l^+(u'_i))<rank'(l^+(v(i)))$, then we reset $v(i)$ to $u'_i$.

If $\alpha(i,j_i-1)> \lambda$, we keep decrementing $j_i$ by one until $\alpha(i,j_i-1)\leq \lambda$. Then, we reset $ncen(i)=ncen(j_i)+1$ and $v(i)$ to $v(j_i)$.

It is not difficult to see that the above algorithm runs in $O(t\log^2
t)$ time, which is $O(m\log^2 m)$ time since $t\leq m$. %Since $t\leq r$ and $r=\log^2 n$, this is $O(t\log\log n)$ time.
\qed
\end{proof}

Since $m\leq r$ and $r=\log^2 n$, we can compute the data structure
for the substem $P$ in $O(r(\log\log n)^2)$ time. The total time for
computing the data structure for all substems of $S$ is $O(n(\log\log
n)^2)$. With these data structures, we show that a feasibility
test can be done in $O(n/r \log^3 r)$ time.

\subsubsection{The Faster Feasibility Test \ftestnew}

Given any $\lambda\in (\lambda_1,\lambda_2)$, the goal is to determine
whether $\lambda$ is feasible. We will work on the stem tree $T_c$, where each node represents a stem of $S$.

Initially, we set $sup(P)=\infty$ and $dem(P)=sup(P)-1$ (so that $dem(P)$ is an infinitely large value but still smaller than $sup(P)$) for every stem $P$ of $T_c$.
We perform a post-order traversal on $T_c$ and maintain a variable $count$, which is the number of centers that have been placed so far. Suppose we are processing a stem $P$. For each child stem $P'$ of $P$, we reset $sup(P)=\min\{sup(P),sup(P')\}$ and $dem(P)=\min\{dem(P),dem(P')\}$. After handling all children of $P$ as above, we process $P$ as follows.

First of all, we increase $count$ by the number of twigs in $P$. Let $V$ be the uncovered vertices of $P$ as defined before, and $v_1',v_2',\ldots,v_t'$ are backbone vertices of $V$.
Note that $t\leq r$. Recall that we have maintained two twig indices $a$ and $b$ for $P$. Depending on whether $sup(P)\leq dem(P)$, there are two main cases.

\paragraph{The case $sup(P)\leq dem(P)$.}
If $sup(P)\leq dem(P)$, then the uncovered vertices in the children of
$P$ can be covered by the center $q$ that determines the value
$sup(P)$ (i.e., $q$ is in a child stem of $P$ and $d(q,v_1)=sup(P)$, where $v_1$ is the lowest backbone vertex of $P$). Note that we do not need
to compute $q$ and we use it only for the discussion.  We do binary search
on the list of $V$ to find the largest index $i\in [1,t]$ such that
$q$ can cover the vertices $V[1,i]$. If no such $i$ exists in $[1,t]$,
then let $i=0$. Such an index $i$ can be found in $O(\log^2r)$
time using the line-constrained 2D sublist LP queries of
Lemma~\ref{lem:sublist}, as shown in the following lemma.

\begin{lemma}\label{lem:140}
Such an index $i$ (i.e., the largest index $i\in [1,t]$ such that
$q$ can cover the vertices $V[1,i]$) can be found in $O(\log^2 r)$ time.
\end{lemma}
\begin{proof}
Given an index $i\in [1,t]$, we show below that we can determine whether $q$
can cover all vertices of $V[1,i]$ in $O(\log r)$ time.

Recall that in our preprocessing, each vertex of $V$ defines two upper
half-planes in $\bbR^2$, and we have built a 2D sublist LP query data
structure on all upper half-planes defined by the vertices of $V$.
Let $q'$ be the point on the $x$-axis of $\bbR^2$ with $x$-coordinate
equal to $-sup(P)$. Let $l$
be the vertical line through $q'$ and let $p$ be the lowest point on
$l$ that is in the common intersection of all upper half-planes
defined by the vertices of $V[1,i]$. An observation is that $q$ can
cover all vertices of $V[1,i]$ if and only if the $y$-coordinate of
$p$ is at most $\lambda$, which can be determined in $O(\log r)$ time
by a line-constrained 2D sublist LP query.

If $q$ can cover all vertices of $V[1,i]$, then we continue the search
on the indices larger than $i$; otherwise, we continue the search on
the indices smaller than $i$. If $q$ cannot cover the vertices of
$V[1,i]$ for $i=1$, then we return $i=0$.
The total time is $O(\log^2 r)$.
\qed
\end{proof}

If $i=t$, then all vertices of $V$ can be covered by $q$. In this case, we reset $sup(P)=\min\{sup(P)+d(v_1,v_m),d(w_b,v_m)-\lambda/w(w_b)\}$, where the latter value is the distance from $v_m$ to the center at the twig $e(v_b,w_g)$.

If $i<t$ (this includes the case $i=0$), then we  increase $i$ by one and reset $count=count+ncen(i)$. If $w(v(i))\cdot d(v(i),v_m)< \lambda$, then we reset $dem(P)=\lambda/w(v(i))-d(v(i),v_m)$ and $sup(P)=d(w_b,v_m)-\lambda/w(w_b)$. Otherwise, we need to place an additional center on $P$ to cover the uncovered vertices of $P$ including $v(i)$, and thus we increase $count$ by one and reset $sup(P)=\min\{d(v(i),v_m)-\lambda/w(v(i)),d(w_b,v_m)-\lambda/w(w_b)\}$ and $dem(P)=\infty$.

\paragraph{The case $sup(P)> dem(P)$.}
In this case, we need to first deal with $dem(P)$, i.e., covering the vertices in the children stems of $P$ that are not covered.

If $dem(P)\geq d(v_1,w_a)-\lambda/w(w_a)$, then the center at the twig $e(v_a,w_a)$ can cover the uncovered vertices in the children stems of $P$. In this case, we increase $count$ by $ncen(1)$. If $w(v(1))\cdot d(v(1),v_m)< \lambda$, then we postpone placing centers to the next stem and reset $dem(P)=\lambda/w(v(1))-d(v(1),v_m)$ and $sup(P)=d(w_b,v_m)-\lambda/w(w_b)$. Otherwise,  we increase $count$ by one and reset $sup(P)=\min\{d(v(1),v_m)-\lambda/w(v(1)),d(w_b,v_m)-\lambda/w(w_b)\}$ and $dem(P)=\infty$.

If $dem(P)< d(v_1,w_a)-\lambda/w(w_a)$, then we do binary search to find the largest index $i\in [1,t]$ such that we can find a center $q$ on the backbone of $P$ to cover all vertices of $V[1,i]$ with $d(v_1,q)\leq dem(P)$. If such an index $i$ does not exit, then we let $i=0$. The following lemma shows that such an index $i$ can be found in $O(\log^3 r)$ time.

\begin{lemma}\label{lem:150}
Such an index $i$ can be found in $O(\log^3 r)$ time.
\end{lemma}
\begin{proof}
Given any index $i$, we first show that we can determine in $O(\log^2 r)$ time the answer to the following question: whether there exists a center $q$ on the backbone of $P$ that can cover all vertices of $V[1,i]$ with  $d(v_1,q)\leq dem(P)$?

By a 2D sublist LP query on the upper half-planes defined by the vertices of $V[1,i]$, we compute the lowest point $p$ in the common intersection of these half-planes. If $y(p)>\lambda$, then the answer to the question is no.
Otherwise, if $x(p)\leq dem(P)$, then the answer to the question is yes.
If $x(p)> dem(P)$, then let $l$ be the vertical line whose $x$-coordinate is $dem(P)$. By a line-constrained 2D sublist LP query, we can compute the lowest point $p'$ on $l$ in the above common intersection of upper half-planes in $O(\log r)$ time. If $y(p')\leq \lambda$, then the answer to the above question is yes; otherwise the answer is no.
The total time to determine the answer to the question is $O(\log^2 r)$.

If the answer is yes, then we continue the search on indices larger than $i$; otherwise we continue on indices smaller than $i$. If the answer to the question is no for $i=1$, then we return $i=0$. The total running time is $O(\log^3 r)$.
\qed
\end{proof}

If $i=t$, then there are two subcases. If $w(v(1))\cdot d(v(1),v_m)< \lambda$ and $dem(P)> d(v_1,v_m)$, then we postpone placing centers to the next stem by resetting $dem(P)=\min\{dem(P)-d(v_1,v_m), \lambda/w(v(1))-d(v(1),v_m)\}$ and $sup(P)=d(w_b,v_m)-\lambda/w(w_b)$. Otherwise, we place a center on the backbone of $P$ of distance $\delta=\max\{d(v(1),v_m)-dem(P),d(v(1),v_m)-\lambda/w(v(1))\}$ from $v_m$. Then, we increase $count$ by one, and reset $sup(P)=\min\{\delta, d(w_b,v_m)-\lambda/w(w_b)\}$ and $dem(P)=\infty$.

If $i\neq t$ (this includes the case $i=0$), then we place a center (at a location determined by the algorithm for Lemma~\ref{lem:150}) to cover $dem(P)$ as well as the vertices of $V[1,i]$ and increase $count$ by one. Next, we increment $i$ by one and increase $count$ by $ncen(i)$. If $w(v(i))\cdot d(v(i),v_m)< \lambda$, then we reset $dem(P)=\lambda/w(v(i))-d(v(i),v_m)$ and $sup(P)=d(w_b,v_m)-\lambda/w(w_b)$. Otherwise, we increase $count$ by one and reset $sup(P)=\min\{d(v(i),v_m)-\lambda/w(v(i)),d(w_b,v_m)-\lambda/w(w_b)\}$ and $dem(P)=\infty$.

This finishes the processing of the stem $P$. After the stem $P_{\gamma}$ that contains the root $\gamma$ is processed, if $sup(P_{\gamma})> dem(P_{\gamma})$, then
we place a center at the root $\gamma$ to cover the uncovered vertices and increase $count$ by one. The value $\lambda$ is feasible if and only if $count\leq k$.
%(note that here $k$ should be the original value from the input).
Since we spend $O(\log^3r)$ time on each stem of $T_c$ and $T_c$ has $O(n/r)$ stems, \ftestnew\ runs in $O(n/r\cdot \log^3 r)$ time. Refer to Algorithm~\ref{algo:ftestnew} for the pseudocode of \ftestnew.

%Let $i_1,i_2,\ldots,i_t$ be the sorted indices in $[1,m]$ such that for each $j\in [1,t]$, either $v_{i_j}$ or $u_{i_j}$ is in $S$.

\begin{algorithm}
\caption{The faster feasibility test algorithm \ftestnew}
\label{algo:ftestnew}
%\SetAlgoNoLine
%\SetAlgoLined
{\small
\KwIn{The stem-tree $T_c$, the original $k$ from the input, and $\lambda\in (\lambda_1,\lambda_2)$}
\KwOut{Determine whether $\lambda$ is feasible}
\BlankLine
%\tcc{All indices below are understood modulo $n$.}
$count\leftarrow 0$\;
\For{\em each stem $P$ of the stem tree $T_c$}
{
  $sup(P)\leftarrow \infty$, $dem(P)\leftarrow sup(P)-1$\;
}
\For{\em each stem $P$ in the post-order traversal of $T_c$}
{
    \For{\em each child $P'$ of $P$}
	{
     $sup(P)=\min\{sup(P),sup(P')\}$,
     $dem(P)=\min\{dem(P),dem(P')\}$\;
	}
    Increase $count$ by the number of twigs of $P$\;
    Let $V$ be the set of the uncovered vertices of $P$, and $v_1',v_2',\ldots,v_t'$ are the backbone vertices of $V$\;
    \eIf{$sup(P)\leq dem(P)$}
    {
      Let $q$ be the center in a child stem of $P$ that gives the value $sup(P)$      \;
      Do binary search to find the largest index $i\in [1,t]$ such that $q$ can cover all vertices of $V[1,i]$\;
      \eIf{$i=t$}
      {
        $sup(P)=\min\{sup(P)+d(v_1,v_m),d(w_b,v_m)-\lambda/w(w_b)\}$\;
      }
      {
         $i++$, $count=count+ncen(i)$\;
         \eIf{$w(v(i))\cdot d(v(i),v_m)< \lambda$}
          {
            $dem(P)=\lambda/w(v(i))-d(v(i),v_m)$, $sup(P)=d(w_b,v_m)-\lambda/w(w_b)$\;
          }
          {
             $count++$, $sup(P)=\min\{d(v(i),v_m)-\lambda/w(v(i)),d(w_b,v_m)-\lambda/w(w_b)\}$, $dem(P)=\infty$\;
          }
      }
    }
    {
      \eIf{$dem(P)\geq d(v_1,w_a)-\lambda/w(w_a)$}
      {
         $count=count+ncen(1)$\;
         \eIf{$w(v(1))\cdot d(v(1),v_m)< \lambda$}
         {
           $dem(P)=\lambda/w(v(1))-d(v(1),v_m)$, $sup(P)=d(w_b,v_m)-\lambda/w(w_b)$\;
         }
         {
           $count++$, $sup(P)=\min\{d(v(1),v_m)-\lambda/w(v(1)),d(w_b,v_m)-\lambda/w(w_b)\}$, $dem(P)=\infty$\;
         }
      }
      {
         Do binary search to find the largest $i\in [1,t]$ such that there exists a center $q$ on the backbone of $P$ to cover all vertices of $V[1,i]$ with $d(v_1,q)\leq dem(P)$\;
         \eIf{$i=t$}
         {
           \eIf{\em $w(v(1))\cdot d(v(1),v_m)< \lambda$ and $dem(P)> d(v_1,v_m)$}
           {
              $dem(P)=\min\{dem(P)-d(v_1,v_m), \lambda/w(v(1))-d(v(1),v_m)\}$,  $sup(P)=d(w_b,v_m)-\lambda/w(w_b)$\;
           }
           {
              $count++$,
              $\delta=\max\{d(v(1),v_m)-dem(P),d(v(1),v_m)-\lambda/w(v(1))\}$, $sup(P)=\min\{\delta, d(w_b,v_m)-\lambda/w(w_b)\}$, $dem(P)=\infty$\;
           }
         }
         {
           $i++$, $count=count+1+ncen(i)$\;
           \eIf{$w(v(i))\cdot d(v(i),v_m)< \lambda$}
           {
             $dem(P)=\lambda/w(v(i))-d(v(i),v_m)$, $sup(P)=d(w_b,v_m)-\lambda/w(w_b)$\;
           }
           {
             $count++$,  $sup(P)=\min\{d(v(i),v_m)-\lambda/w(v(i)),d(w_b,v_m)-\lambda/w(w_b)\}$, $dem(P)=\infty$\;
           }
         }
      }
    }

}
\If{$sup(P_{\gamma})>dem(P_{\gamma})$}
{
        $count++$\;
}
Return true if and only if $count\leq k$\;
%\eIf{$count\leq k$}
%{
%        return true\tcc*{$\gamma$ is feasible}
%}
%{
%        return false\tcc*{$\gamma$ is not feasible}
%}
}
\end{algorithm}

\subsection{Phase 2}
In this phase, we will finally compute the optimal objective value $\lambda^*$, using the faster feasibility test \ftestnew. Recall that we have computed a range $(\lambda_1,\lambda_2]$ that contains $\lambda^*$ after Phase 1.

We first form a stem-partition of $T$. While there is more than one leaf-stem, we do the following.
Let $S$ be the set of all leaf-stems.
%For each leaf-stem, we compute the lines as in Section~\ref{sec:first}. Let $L$ be the set of all lines for all leaf-stems of $S$, i.e., we put all these lines in the same $xy$-coordinate system.
For each stem $P\in S$, we compute the set of lines as in Section~\ref{sec:first}, and let $L$ be the set of the lines for all stems of $S$.
%Consider the line arrangement $\calA(L)$.
With Lemma~\ref{lem:arrangement} and \ftestnew, we compute the two vertices $v_1(L)$ and $v_2(L)$ of the arrangement $\calA(L)$ as defined in Section~\ref{sec:pre}. We update $\lambda_1=\max\{\lambda_1,y(v_2(L))\}$ and $\lambda_2=\min\{\lambda_2,y(v_1(L))\}$. As discussed in Phase 1, each stem $P$ of $S$ does not have any active values (in the matrices defined by $P$).
Next, for each stem $P$ of $S$, we perform the post-processing procedure as in Section~\ref{sec:phase0}, i.e., place centers on $P$, subtract their number from $k$, and replace $P$ by attaching a twig or a thorn to its top vertex. Let $T$ be the modified tree.

%For each leaf-stem, we form matrices by Lemma~\ref{lem:matrixform}. Let $\calM$ denote the set of matrices for all leaf-stems. We call \msearch\ on $\calM$, with stopping count $c=0$, using \ftestnew. When \msearch\ stops, we have an updated range $(\lambda_1,\lambda_2)$. Since $c=0$, none of the matrix element of $\calM$ is active. Next, for each leaf-stem $P$, we perform the post-processing procedure as in Section~\ref{sec:phase0}, i.e., place centers on $P$, subtract their number from $k$, and replace $P$ by attaching a twig or a thorn to its top vertex. Let $T$ be the modified tree.

After the while loop, $T$ is a single stem. Then, we apply above algorithm on the only stem $T$, and the obtained value $\lambda_2$ is $\lambda^*$.
%Then, we apply \msearch\ on the matrices formed by the stem with $c=0$. After  \msearch\ stops, the value $\lambda_2$ is $\lambda^*$.
The running time of Phase 2 is bounded by $O(n\log n)$, which is analyzed in the following theorem.

\begin{theorem}\label{theo:kcenter}
The $k$-center problem on $T$ can be solved in $O(n\log n)$ time.
\end{theorem}
\begin{proof}
As discussed before, Phases 0 and 1 run in $O(n\log n)$ time. Below we focus on Phase 2.

First of all, as in~\cite{ref:FredericksonPa91}, the number of iterations of the while loop is $O(\log n)$ because the number of leaf-stems is halved after each iteration. In each iteration, let $n'$ denote the total number of backbone vertices of all leaf-stems in $S$. Hence, $|L|=O(n')$. Thus, the call to Lemma~\ref{lem:arrangement} with \ftestnew\ takes $O((n'+n/r\cdot \log^3 r)\log n')$ time. The total time of the post-processing procedure for all leaf-stems of $S$ is $O(n')$. Since all leaf-stems of $S$ will be removed in the iteration,
%Since the backbone vertices of these leaf-stems will not be backbone vertices of any leaf-stems in the subsequent iterations,
the total sum of all such $n'$ is $O(n)$ in Phase 2.
Therefore, the total time of the algorithm in Lemma~\ref{lem:arrangement} in Phase 2 is $O(n\log n + n/r\cdot \log^3 r\log^2 n)$, which is $O(n\log n)$ since $r=\log^2 n$. Also, the overall time for the post-processing procedure in Phase 2 is $O(n)$.
Therefore, the total time of Phase 2 is $O(n\log n)$. This proves the theorem.
\qed
\end{proof}

The pseudocode in Algorithm~\ref{algo:kcenter} summarizes the overall algorithm.

\begin{algorithm}
\caption{The $k$-center algorithm}
\label{algo:kcenter}
%\SetAlgoNoLine
%\SetAlgoLined
\KwIn{A tree $T$ and an integer $k$}
\KwOut{The optimal objective values $\lambda^*$ and $k$ centers in $T$}
\BlankLine
Perform the preprocessing in Section~\ref{sec:preprocess} and compute the ``ranks'' for all vertices of $T$\;
\tcc{Phase 0}
$r\leftarrow \log^2 n$\;
Form a stem-partition of $T$\;
\While{\em there are more than $2n/r$ leaves in $T$}
{
   Let $S$ be the set of all leaf-stems of lengths at most $r$\;
   Form the set $\calM$ of matrices for all leaf-stems of $S$ by Lemma~\ref{lem:matrixform}\;
   Let $n'$ be the total number of all backbone vertices of the leaf-stems of $S$\;
   Call \msearch\ on $\calM$ with stopping count $c=n'/(2r)$, using \ftest0\;
   \For{\em each leaf-stem $P$ of $S$ with no active values}
   {
     Perform the post-processing on $P$, i.e., place centers on $P$, subtract their number from $k$, replace $P$ by a thorn or a twig, and modify the stem-partition of $T$\;
   }
}
\tcc{Phase 1}
For a stem-partition of $T$, and for each stem, partition it into substems of lengths at most $r$\;
Let $S$ be the set of all substems, and form the stem-tree $T_c$\;
Compute the set $L$ of lines for all stems of $S$ in the way discussed in Section~\ref{sec:first}\;
Compute the two vertices $v_1(L)$ and $v_2(L)$ of $\calA(L)$ by Lemma~\ref{lem:arrangement} and \ftest0, and update $\lambda_1$ and $\lambda_2$\;
\For{\em each substem $P$ of $S$}
{
  Compute the data structure for the faster feasibility test \ftestnew\;
}
\tcc{Phase 2}
Form a stem-partition of $T$\;
\While{\em there is more than one leaf-stem in $T$}
{
  %Form the set $\calM$ of sorted matrices for all leaf-stems\;
%  Call \msearch\ on $\calM$ with stopping count $c=0$, using \ftestnew\ for $T_c$\;
  Compute the set $L$ of lines for all stems of $S$ in the way discussed in Section~\ref{sec:first}\; \label{ln:lineset10}
Compute the two vertices $v_1(L)$ and $v_2(L)$ of $\calA(L)$ by Lemma~\ref{lem:arrangement} and \ftestnew, and update $\lambda_1$ and $\lambda_2$\;
\label{ln:vertex10}
  \For{\em each leaf-stem of $S$}
  {
      Perform the post-processing on $P$, i.e., place centers on $P$, subtract their number from $k$, replace $P$ by a thorn or a twig, and modify the stem-partition of $T$\;
  }
}
Compute the set $L$ of the lines for the only leaf-stem $T$\; \label{ln:lineset20}
Compute the two vertices $v_1(L)$ and $v_2(L)$ of $\calA(L)$ by Lemma~\ref{lem:arrangement} and \ftestnew, and update $\lambda_1$ and $\lambda_2$\;
\label{ln:vertex20}
$\lambda^*=\lambda_2$\;
Apply \ftest0\ on $\lambda=\lambda^*$ to find $k$ centers in the original tree $T$\;
\end{algorithm}

\section{The Discrete $k$-Center Problem}
\label{sec:discrete}

In this section, we extend our techniques to solve in $O(n\log n)$ time the discrete $k$-center problem on $T$ where centers must be located at the vertices of $T$. In fact, the problem becomes easier due to the following observation.

\begin{observation}\label{obser:discrete}
The optimal objective value $\lambda^*$ is equal to $w(v)\cdot d(v,u)$ for two vertices $u$ and $v$ of $T$ (i.e., a center is placed at $u$ to cover $v$).
\end{observation}

The previous $O(n\log^2 n)$ time algorithm in~\cite{ref:MegiddoAn81}  relies on this observation. Megiddo et al.~\cite{ref:MegiddoAn81} first computed in $O(n\log^2 n)$ time a collection of $O(n\log n)$ sorted subsets that contain the intervertex distances $d(u,v)$ for all pairs $(u,v)$ of vertices of $T$. By multiplying the weight $w(v)$ by the elements in the subsets corresponding to each vertex $v$, $\lambda^*$ is contained in these new sorted subsets. Then, $\lambda^*$ can be computed in $O(n\log^2 n)$ time by searching these sorted subsets, e.g., using \msearch. Frederickson and Johnson~\cite{ref:FredericksonFi83} later proposed an $O(n\log n)$-time algorithm that computes a succinct representation of all intervertex distances of $T$ by using sorted Cartesian matrices. With \msearch, their algorithm solves the unweighted case of the problem in $O(n\log n)$ time. However, their techniques may not be generalized to solving the weighted case because once we multiply the vertex weights by the elements of those Cartesian matrices, the new matrices are not sorted any more (i.e., we cannot guarantee that both columns and rows are sorted because different rows or columns are multiplied by weights of different vertices).

Our algorithm uses similar techniques as those for the previous non-direcrete $k$-center problem. In the following we briefly discuss it and mainly focus on pointing out the differences.

First of all, we need to modify the feasibility test algorithm \ftest0.
The only difference is that when $sup(u)\leq dem(u)$ and $dem(u)<d(u,v)$, instead of placing a center in the interior of the edge $e(u,v)$, we place a center at $u$ and update $sup(v)=\min\{sup(v),d(u,v)\}$ (i.e., use this to replace Line~\ref{ln:ftest0} in the pseudocode of Algorithm~\ref{algo:ftest0}). The running time is still $O(n)$. We use \dftest0\ to denote the new algorithm.

%The new algorithm, denoted by \dftest0, is given in Algorithm~\ref{algo:dftest0}.

\subsection{The Algorithm for Stems}
\label{sec:dstem}

%Our algorithm for $T$ will also maintain an interval $(\lambda_1,\lambda_2]$ that contains $\lambda^*$. Since we know that $\lambda_2$ is feasible, any feasibility test will be on a value $\lambda\in (\lambda_1,\lambda_2)$.

The stem is defined slightly differently than before. Suppose we have a range $(\lambda_1,\lambda_2]$ that contains $\lambda^*$.
Each backbone vertex $v_i$ of a stem $P$ still has at most one thorn and one twig. The thorn $e(u_i,v_i)$ is defined in the same way as before. However, a twig now consists of two edges $e(v_i,b_i)$ and $e(b_i,w_i)$ such that $w(w_i)\cdot d(w_i,v_i)\geq \lambda_2$ and $w(w_i)\cdot d(w_i,b_i)\leq \lambda_1$, which means that we have to place a center at $b_i$ to cover $w_i$ under any $\lambda\in (\lambda_1,\lambda_2)$. We still call $w_i$ a twig vertex, and following the terminology in~\cite{ref:FredericksonPa91}, we call $b_i$ a {\em bud}.

Next we give an algorithm to solve the $k$-center problem on a stem $P$ of $m$ backbone vertices. The algorithm is similar to that in Section~\ref{sec:second} and uses \msearch, but we use a different way to form matrices based on Observation~\ref{obser:discrete}. (Note that we will not need a similar algorithm as that in Section~\ref{sec:first}.) Let $\lambda^*$ temporarily refer to the optimal objective value for the $k$-center problem on $P$ in this subsection, and we assume $\lambda^*\in (\lambda_1,\lambda_2]$.

Let $v_1,v_2,\ldots,v_m$ be the backbone vertices of $P$. We again assume that all backbone vertices are in the $x$-axis such that $v_1$ is at the origin and $v_i$ has $x$-coordinate $d(v_1,v_i)$. As in Section~\ref{sec:first}, for each thorn vertex $u_i$, we define two points $u_i^l$ and $u_i^r$ on the $x$-axis (whose weights are equal to $w(u_i)$), and we do the same for each bud and each twig vertex. Let $\calP$ be the set of all vertices on the $x$-axis. Hence, $|\calP|\leq 7m$.

We sort all vertices of $\calP$ from left to right, and let the sorted list be $z_1,z_2,\ldots,z_t$ for $t\leq 7m$. For any $z_i$, we use $x(z_i)$ to denote its $x$-coordinate.
For each vertex $z_i$, we define two sorted arrays $A_i^l$ and $A_i^r$ of lengths at most $t$ as follows. For each $j\in [1,t-i+1]$, define $A_i^r[j]=w(z_i)\cdot (x(z_{t+1-j})-x(z_i))$. For each $j\in [1,i]$, define $A_i^l[j]=w(z_i)\cdot (x(z_i)-x(z_j))$. Both arrays are sorted.

Let $\calM$ denote the set of all $O(m)$ sorted arrays defined above. By Observation~\ref{obser:discrete}, $\lambda^*$ must be an element of an array in $\calM$. With $O(m\log m)$ time preprocessing, each array element of $\calM$ can be computed in $O(1)$ time. By applying \msearch\ on $\calM$ with stopping count $c=0$ and using \dftest0, we can compute $\lambda^*$ in $O(m\log m)$ time (i.e., \msearch\ produces $O(\log m)$ values for feasibility tests in $O(m\log m)$ time).

\subsection{Solving the Problem in $T$}

In the sequel, we solve the discrete $k$-center problem in $T$. First, we do the same preprocessing as in Section~\ref{sec:preprocess}. Then, we have three phases as before. Let $r=\log^2 n$.

\subsubsection{Phase 0}

We assume that $T$ has more than $2n/r$ leaves since otherwise we could skip this phase. Phase 0 is the same as before except the following changes. First, we use \dftest0\ to replace \ftest0. Second, we form the matrix set $\calM$ in the way discussed in Section~\ref{sec:dstem}. Third, for each leaf-stem $P$ without active values, we modify the post-processing procedure as follows.

Let $z$ be the top vertex of $P$. We run \dftest0\ on $P$ with $z$ as the root and $\lambda=\lambda'$ as any value in $(\lambda_1,\lambda_2)$. As before, after $z$ is processed, depending on whether $sup(z)\leq dem(z)$, there are two main cases.

\paragraph{The case $sup(z)\leq dem(z)$.}
If $sup(z)\leq dem(z)$, we define $q$ and $V(q)$ in the same way as before but $V(q)$ should not include bud vertices. A difference is that $q$ now is a vertex of $P$. Note that $q\neq z$ because in this case we do not need to place a center at $z$.

%Let $u$ be the vertex of $V(q)$ with the largest rank. Based on the following lemma, whose proof is similar in spirit to that of Lemma~\ref{lem:60}, we replace $P$ be a twig $e(u,q)\cup e(q,z)$.
%
%\begin{lemma}
%\begin{enumerate}
%\item
%$q$ is on $\pi(u,z)$, i.e., $u$ is a descendent of $q$.
%\item
%Let $v$ be the vertex on $\pi(u,z)$ closest to $z$ such that $w(u)\cdot d(u,v)\leq \lambda$ for any $\lambda\in (\lambda_1,\lambda_2)$. Then, $q$ is $v$. \end{enumerate}
%\end{lemma}

Let $w$ be the vertex that makes $q$ as a center. Specifically, refer to the pseudocode in Algorithm~\ref{algo:ftest0}, where we place a center $q$ at vertex $u$ at Line 10. Let $w$ be the vertex that determines the value $dem(u)$, i.e., $dem(u)=\lambda'/w(w)-d(w,u)$. Note that $w$ must be a descendent of $u$ ($=q$).
In this case, we replace $P$ by a twig consisting of two edges $e(z,q)$ and $e(q,w)$ with lengths equal to $d(z,q)$ and $d(q,w)$, respectively.
To find the vertex $w$, one way is to modify \dftest0\ so that the vertex that determines the value $dem(v)$ for each vertex $v$ of $T$ is also maintained.
Another way is that $w$ is in fact the vertex of $V(q)$ with the largest rank, which is proved in the following lemma (whose proof is similar to that of Lemma~\ref{lem:60}). Note that this is actually consistent with our  way for creating twigs in the previous non-discrete case.

\begin{lemma}\label{lem:160}
$w$ is the vertex of $V(q)$ with the largest rank.
\end{lemma}
\begin{proof}
Let $v$ be any vertex of $V(q)\setminus\{w\}$. Our goal is to show that $rank(w)>rank(v)$.
Note that $q$ is either a backbone vertex or a bud.

We first discuss the case where $q$ is a backbone vertex.
We define $V_1(q)$ to be the set of vertices of $V(q)$ in the subtree rooted at $q$ and let  $V_2(q)=V(q)\setminus V_1(q)$. Depending on whether $v$ is in $V_1(q)$ or $V_2(q)$, there are two subcases.

\begin{figure}[t]
\begin{minipage}[t]{\linewidth}
\begin{center}
\includegraphics[totalheight=1.3in]{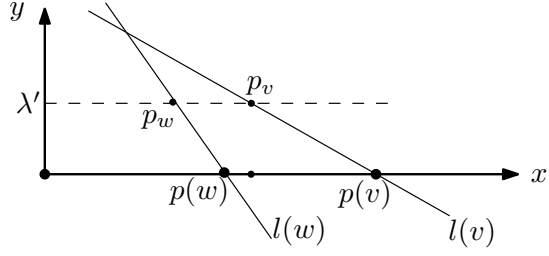}
\caption{\footnotesize Illustrating the proof of Lemma~\ref{lem:160}. Note that $p(v)$ and $p(w)$ are the points defined respectively by $v$ and $w$ in the preprocessing of Section~\ref{sec:preprocess}.}
\label{fig:rank25}
\end{center}
\end{minipage}
\vspace*{-0.15in}
\end{figure}

If $v\in V_1(q)$, assume to the contrary that $rank(w)<rank(v)$. Recall that $v$ defines a line $l(v)$ and $w$ defines a line $l(w)$ in $\bbR^2$ in the preprocessing (see Section~\ref{sec:preprocess}). Refer to Fig.~\ref{fig:rank25}. Let $p_v$ and $p_w$ denote the intersections of the horizontal line $y=\lambda'$ with $l(v)$ and $l(w)$, respectively. Since $\lambda'\in(\lambda_1,\lambda_2)$, by the definition of ranks, $x(p_w)<x(p_v)$. Note that $x(p_v)$ corresponds to a point $q_v$ in $\pi(v,\gamma)$ in the sense that $x(p_v)=d(\gamma,q_v)$, and $q_v$ is actually the point on $\pi(v,\gamma)$ closest to $\gamma$ that can cover $v$. Similarly,  $x(p_w)$ corresponds to a point $q_w$ in $\pi(w,\gamma)$. Since $x(p_w)<x(p_v)$, $d(\gamma,q_w)<d(\gamma,q_v)$. One can verify that this contradicts with that $w$ is the vertex that makes $q$ as a center (i.e., $w$ determines the value $dem(u)$), because both $v$ and $w$ are descendants of $q$.

If $v\in V_2(q)$, first note that $v$ cannot be a twig vertex since otherwise $q$ would need to be a bud in order to cover $v$. Depending on whether $v$ is a backbone vertex or a thorn vertex, there are two subcases.

\begin{figure}[h]
\begin{minipage}[t]{0.49\linewidth}
\begin{center}
\includegraphics[totalheight=1.3in]{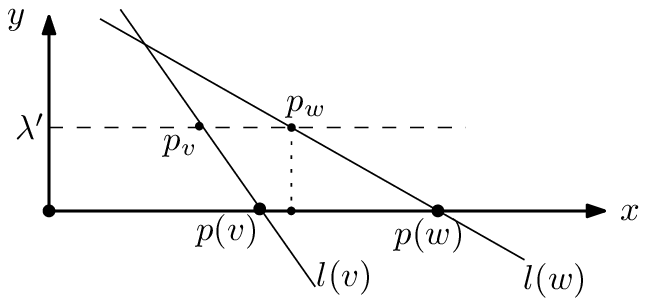}
\caption{\footnotesize Illustrating the proof of Lemma~\ref{lem:160}. }
\label{fig:rank3}
\end{center}
\end{minipage}
\hspace{0.02in}
\begin{minipage}[t]{0.49\linewidth}
\begin{center}
\includegraphics[totalheight=1.3in]{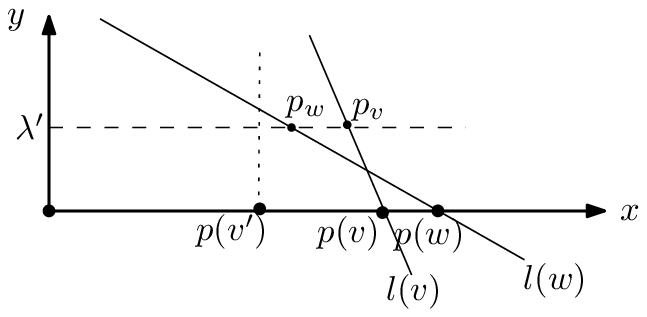}
\caption{\footnotesize Illustrating the proof of Lemma~\ref{lem:160}.}
\label{fig:rank4}
\end{center}
\end{minipage}
\vspace*{-0.15in}
\end{figure}

\begin{enumerate}
\item
If $v$ is a backbone vertex, then $v$ is an ancestor of $q$.
%recall that $v$ defines a line $l(v)$ and $w$ defines a line $l(w)$ in $\bbR^2$ in the preprocessing (see Section~\ref{sec:preprocess}). Refer to Fig.~\ref{fig:rank3}.
Let $p_v$ and $p_w$ be the intersections of the horizontal line $y=\lambda'$ with $l(v)$ and $l_w$, respectively (e.g., see Fig.~\ref{fig:rank3}).
Since $w$ determines the center $q$ and $v$ is an ancestor of $q$, according to \dftest0, it holds that $w(w)\cdot d(w,v)<\lambda'$.
Since $q$ is an ancestor of $w$, $v$ is also an ancestor of $w$.
Therefore, the point $p(v)$ (i.e., the intersection of $l(v)$ with the $x$-axis) must be to the left of the vertical line through $p_w$. Since the slope of $l(v)$ is not positive, $x(p_v)<x(p_w)$. Because $\lambda'\in (\lambda_1,\lambda_2)$, we obtain that $rank(v)<rank(w)$.

\item
If $v$ is a thorn vertex, assume to the contrary that $rank(w)<rank(v)$. Let $v'$ be the backbone vertex that connects $v$. Since $v\in V_2(q)$ and $q$ is a backbone vertex, $v'$ must be an ancestor of $q$. Because $rank(w)<rank(v)$, in the following we will prove $w(v)\cdot d(v,q)>\lambda'$, which implies that $q$ cannot cover $v$ and thus incurs contradiction.

By our preprocessing in Section~\ref{sec:preprocess}, the vertex $v'$ defines a point $p(v')$ in the $x$-axis of $\bbR^2$ (e.g., see Fig.~\ref{fig:rank4}). With a little abuse of notation, we use $x(v')$ to denote the $x$-coordinate of $p(v')$.  We define $p_v$ and $p_w$ in the same way as before.

Since $w$ determines $q$ and $v'$ is an ancestor of $q$,
according to \dftest0, it holds that $w(w)\cdot d(w,v')<\lambda'$. This implies that $x(v')<x(p_w)$. Since $rank(w)<rank(v)$ and $\lambda'\in (\lambda_1,\lambda_2)$, $x(p_w)<x(p_v)$. Thus, $x(v')<x(p_v)$ and the $y$-coordinate of the intersection of the vertical line through $p(v')$ and $l(v)$ is larger than $\lambda'$. Note that the above $y$-coordinate is equal to $w(v)\cdot d(v,v')$. Therefore, we obtain $w(v)\cdot d(v,v')>\lambda'$. Since the path $\pi(q,v)$ contains the thorn $e(v,v')$, we have $d(v,q)\geq d(v,v')$. Thus, it holds that $w(v)\cdot d(v,q)>\lambda'$.
\end{enumerate}

The above proves that $rank(w)>rank(v)$ for the case where $q$ is a backbone vertex.

Next we discuss the case where $q$ is a bud. Since $w$ is a descendent of $q$, $w$ must be a twig vertex on the same twig as $q$. Let $v'$ be the backbone vertex that connects $q$. We say that $v$ is {\em above} $v'$ if $\pi(z,v')$ either contains $v$ (when $v$ is a backbone vertex) or contains the backbone vertex that connects $v$ (when $v$ is a thorn vertex); otherwise, $v$ is {\em below} $v'$.
If $v$ is above $v'$, then the analysis is similar to the above subcase $v\in V_2(q)$. We omit the details. In the following, we analyze the case where $v$ is below $v'$. Although $v$ may be either a backbone vertex or a thorn vertex, we prove $rank(w)>rank(v)$ in a uniform way.

Assume to the contrary that $rank(w)<rank(v)$. Again, refer to Fig.~\ref{fig:rank4}. We define the points in the figure in the same way as before except that $v'$ is now the backbone vertex that connects $q$. Our goal is to show that $w(v)\cdot d(v,q)>\lambda'$, which incurs contradiction since $q$ covers $v$. To this end, since $d(v,v')\leq d(v,q)$, it is sufficient to show that $w(v)\cdot d(v,v')>\lambda'$.
The proof is similar to the above subcase where $v$ is a thorn vertex, and we briefly discuss it below.

%Since $rank(w)<rank(v)$, $x(p_w)<x(p_v)$. Since $x(v')<x(q)=x(p_w)$, the $y$-coordinate of the intersection of the vertical through $v'$ and $l(v)$ is larger than $\lambda'$, and note that the above $y$-coordinate is exactly $w(v)\cdot d(v,v')$. Thus, we obtain  $w(v)\cdot d(v,v')>\lambda'$.

Since $w$ determines $q$ and $v'$ is an ancestor of $q$,
it holds that $w(w)\cdot d(w,v')<\lambda'$. This implies that $x(v')<x(p_w)$. Since $rank(w)<rank(v)$ and $\lambda'\in (\lambda_1,\lambda_2)$, $x(p_w)<x(p_v)$. Thus, $x(v')<x(p_v)$ and the $y$-coordinate of the intersection of the vertical line through $p(v')$ and $l(v)$ is larger than $\lambda'$. Again, the above $y$-coordinate is equal to $w(v)\cdot d(v,v')$. Therefore, $w(v)\cdot d(v,v')>\lambda'$.

This proves the lemma.
\qed
\end{proof}

%Note that $w(w)\cdot d(w,z)>\lambda$ since otherwise we would not need to place the center $q$ at $u$ to cover $u'$ (indeed, we place the center $q$ at $u$ because the condition in Line 9 is satisfied). Clearly, $w(u')\cdot d(u',q)\leq \lambda$. Since the matrices formed by $P$ do not have any active values and $\lambda\in (\lambda_1,\lambda_2)$, $w(u')\cdot d(u',z)>\lambda$ implies that $w(u')\cdot d(u',z)\geq \lambda_2$ and $w(u')\cdot d(u',q)\leq \lambda$ implies that $w(u')\cdot d(u',q)\leq \lambda_1$. Therefore, $e(z,q)$ and $e(q,u')$ indeed constitute a twig for $z$.

The above replaces $P$ by attaching a twig $e(z,q)\cup e(q,w)$ to $z$.
In addition, if $z$ already has another twig with bud $q'$, then we discard the new twig if $d(q,z)\geq d(q',z)$ and discard the old twig otherwise. This guarantees that $z$ has at most one twig.

%Due to that $\lambda^*\in (\lambda_1,\lambda_2]$, we obtain $w(u')\cdot d(u',z)> \lambda^*$ and $w(u')\cdot d(u',q)\leq \lambda^*$.

\paragraph{The case $sup(z)>dem(z)$.}
If $sup(z)> dem(z)$, then we define $V$ in the same way as before but excluding the buds. Let $u$ be the vertex of $V$ with the largest rank. As before in Lemma~\ref{lem:70}, $u$ dominates all other vertices of $V$ and thus we replace $P$ by a thorn $e(u,z)$ whose length is equal to $d(u,z)$.
%Because a center at $z$ would cover $u$, $w(u)\cdot d(u,z)\leq \lambda$ holds. Since $P$ does not have any active values, this means that $w(u)\cdot d(u,z)\leq \lambda_1$. Therefore, $e(u,z)$ is indeed a thorn.
In addition, if $z$ already has another thorn $e(z,u')$, then as before (and for the same reason), we discard one of $u$ and $u'$ whose rank is smaller.

This finishes the post-processing procedure on $P$. The running time is still $O(m)$.

By similar analysis as in Lemma~\ref{lem:90}, Phase 0 still runs in $O(n\log n)$ time and we omit the details.

\subsubsection{Phase 1}

First of all, we still form a stem-tree $T_c$ as before and each node represents a substem of length at most $r$. Instead of using the line arrangement searching technique, we now resort to \msearch. Let $S$ be the set of all substems. Let $\calM$ be the set of matrices of all these substems formed in the way described in Section~\ref{sec:dstem}. We apply \msearch\ on $\calM$ with stopping count $c=0$ and using \dftest0. Since $\calM$ has $O(n)$ arrays of lengths $O(r)$, \msearch\ will produce $O(\log n)$ values for feasibility tests in $O(n\log r)$ time. The total feasibility test time is $O(n\log n)$. Since $c=0$, after \msearch\ stops, we have an updated range $(\lambda_1,\lambda_2)$ and no matrix element of $\calM$ is active.  Let $\lambda$ be an arbitrary value in $(\lambda_1,\lambda_2)$.

We will compute a data structure for each stem $P$ of $S$ so that the feasibility test can be made in sublinear time. We will show that after $O(n\log n)$ time preprocessing, each feasibility test can be done in $O(n/r\log^3 r)$ time.
%In the following, we first compute the data structure and then give the faster feasibility test algorithm, denoted by \dftestnew.
Let $P$ be a substem with backbone vertices $v_1,v_2,\ldots,v_m$, with $v_m$ as the top vertex.
The preprocessing algorithm works in a similar way as before.

\paragraph{The cleanup procedure.}
Again, we first perform a {\em cleanup procedure} to remove all vertices of $P$ that can be covered by the centers at the buds of the twigs.
Note that all buds and twig vertices can be automatically covered by the centers at buds. So we only need to find those  backbone and thorn vertices that can be covered by buds.
This is easier than before because the locations of the centers are now fixed at buds.

We use the following algorithm to find all backbone and thorn vertices that can be covered by buds to their left sides. Let $i'$ be the smallest index such that $v_{i'}$ has a bud $b_{i'}$. The algorithm maintains an index $t$. Initially, $t=i'$. For each $i$ incrementally from $i'$ to $m$, we do the following. If $v_i$ has a bud $b_i$, then we reset $t$ to $i$ if $d(b_i,v_i)\leq d(b_t,v_i)$, because for any $j\geq i$ such that $v_j$ (resp., $u_j$) is covered by $b_t$, it is also covered by $b_i$. Next, if $w(v_i)\cdot d(v_i,b_t)\leq \lambda$, then we mark $v_i$ as ``covered''. If $u_i$ exits and $w(u_i)\cdot d(u_i,b_t)\leq \lambda$, then we mark $u_i$ as ``covered''. Note that although $u_i$ is not an ancestor or a descendent of $b_t$,  we can still compute $d(u_i,b_t)$ in $O(1)$ time because $d(u_i,b_t)=d(u_i,v_i)+d(v_i,b_t)$, and the latter two distances can be computed in $O(1)$ time due to the preprocessing in Section~\ref{sec:preprocess}.

The above algorithm runs in $O(m)$ time. In a symmetric way by scanning $P$ from right to left, we can also mark all backbone and thorn vertices that are covered by buds to their right sides. This finishes the cleanup procedure. Again, we define $V$ in the same way as before, and let $v_1',v_2',\ldots,v_t'$ be the backbone vertices in $V$. Also, define $V[i,j]$ in the same way as before.

\paragraph{Computing the data structure.}
In the sequel, we compute a data structure for $P$.

As before, we maintain an index $a$ of a twig such that $d(v_1,b_a)\leq d(v_1,b_i)$ for any other twig index $i$, and maintain an index $b$ of a twig such that $d(v_m,b_b)\leq d(v_m,b_i)$ for any other twig index $i$. Both $a$ and $b$ can be found in $O(m)$ time.

For each $i\in [1,t]$, we will also compute an integer $ncen(i)$ and a vertex $v(i)$ of $V$, whose definitions are similar as before. In addition, we maintain another vertex $q(i)$. The details are given below.
Define $j_i$ similarly as before, i.e., it is the smallest index in $[i,t]$ such that it is not possible to cover all vertices in $V[i,j_i]$ by a center located at a vertex of $P$ under $\lambda$. Again, if such an index does not exist in $[i,t]$, then let $j_i=t+1$.

If $j_i=t+1$, then we define $ncen(i)=0$, and define $v(i)$ as the vertex in $V[i,t]$ with the largest rank. We should point out our definition on $v(i)$ is consistent with before, which was based on $rank'(l^+(v))$ for $v\in V$, because for any two vertex $v$ and $v'$ in $V$, $rank(v)>rank(v')$ if and only if $rank'(l^+(v))<rank'(l^+(v'))$.
With the same analysis as before, $v(i)$ dominates all other vertices of $V[i,t]$.
We define $q(i)$ as follows.
If $w(v(i))\cdot d(v(i),v_m)<\lambda$, then $q(i)$ is undefined. Otherwise, $q(i)$ is the largest index $j\in [i,t]$ such that $w(v(i))\cdot d(v(i),v_j)\leq \lambda$. Therefore, $q(i)$ refers to the index of the backbone vertex closest to $v_m$ that can cover $v(i)$, and by the definition of $v(i)$, $q(i)$ can also cover all other vertices of $V[i,t]$ (the proof is similar to Lemma~\ref{lem:120} and we omit it).

If $j_i<t+1$, then define $ncen(i)=ncen(j_i)+1$, $v(i)=v(j_i)$, and $q(i)=q(j_i)$.

To compute $ncen(i)$, $v(i)$, and $q(i)$ for all $i\in [1,t]$, observe that once $v(i)$ is known, $q(i)$ can be computed in additional $O(\log m)$ time by binary search on the backbone vertices of $P$. Therefore, we will focus on computing $ncen(i)$ and $v(i)$. We use a similar algorithm as that in Lemma~\ref{lem:130}. To this end, we need to solve the following subproblem: Given any two indices $i\leq j$ in $[1,t]$, we want to compute the optimal objective value, denoted by $\alpha'(i,j)$, of the discrete one-center problem on the vertices of $V[i,j]$. This can be done in $O(\log^2 m)$ time by using the 2D sublist LP query data structure, as shown in the following lemma.

\begin{lemma}\label{lem:160}
With $O(m\log m)$ time preprocessing, we can compute $\alpha'(i,j)$ in $O(\log^2 m)$ time for any two indices $i\leq j$ in $[1,t]$.
\end{lemma}
\begin{proof}
As preprocessing, in $O(t\log t)$ time we build the 2D sublist LP query data structure on the upper half-planes defined by the vertices of $V$ in the same way as before.

Given any indices $i\leq j$ in $[1,t]$, by a 2D sublist LP query, we compute in $O(\log^2 t)$ time the lowest point $p$ in the common intersection $C$ of the upper half-planes defined by the vertices of $V[i,j]$. Let $q$ be the point on the backbone of $P$ corresponding to the $x$-coordinate of $p$ (i.e., $d(v_1,q)=x(p)$). Note that $q$ is essentially the optimal center for the non-discrete one-center problem on the vertices of $V[i,j]$. But since we are considering the discrete case, the optimal center $q'$ for the discrete problem can be found as follows. If $q$ is located at a vertex of $P$, then $q'=q$ and $\alpha'(i,j)$ is equal to the $y$-coordinate of $p$. Otherwise, let $v$ and $v'$ be the two backbone vertices of $P$ immediately on the left and right sides of $q$, respectively. Then, $q'$ is either $v$ or $v'$, and this is because the boundary of $C$, which is the upper envelope of the bounding lines of the upper half-planes defined by  the vertices of $V[i,j]$, is convex. To compute $\alpha'(i,j)$, we do the following. Let $l_v$ be the line in $\bbR^2$ whose $x$-coordinate is equal to $d(v_1,v)$. Defined $l_{v'}$ similarly. Let $p_v$ be the lowest point of $l_v$ in $C$. Define $p_{v'}$ similarly. Then, $\alpha'(i,j)$ is equal to the $y$-coordinate of the lower point of $p_v$ and $p_{v'}$, and the center $q'$ can also be determined correspondingly. Both $v$ and $v'$ can be found by binary search on the backbone vertices of $V$ in $O(\log m)$ time. The point $p_v$ (resp., $p_{v'}$) can be found by a line-constrained 2D sublist LP query in $O(\log t)$ time.

Since $t\leq m$, we can compute $\alpha'(i,j)$ in $O(\log^2 m)$ time.
\qed
\end{proof}

With the preceding lemma, we can use a similar algorithm as in Lemma~\ref{lem:130} to compute $ncen(i)$ and $v(i)$ for all $i\in [1,t]$ in $O(m\log^2 m)$ time. Again, $q(i)$ for all $i\in [1,t]$ can be computed in additional $O(m\log m)$ time  by binary search.

Recall that $m\leq r$.
Hence, the preprocessing time for $P$ is $O(r\log^2 r)$, which is $O(r(\log\log n)^2)$ time since $r=\log^2 n$. The total time for computing the data structure for all substems of $S$ is $O(n(\log\log n)^2)$. With these data structures, we show that a feasibility test can be done in $O(n/r \log^3 r)$ time.

\paragraph{The faster feasibility test {\em \dftestnew}.}
The algorithm is similar as before, and we only explain the differences by referring to the pseudocode given in Algorithm~\ref{algo:dftestnew} for \dftestnew.

\begin{algorithm}
\caption{The faster feasibility test algorithm \dftestnew}
\label{algo:dftestnew}
%\SetAlgoNoLine
%\SetAlgoLined
\KwIn{The stem-tree $T_c$, the original $k$ from the input, and $\lambda\in (\lambda_1,\lambda_2)$}
\KwOut{Determine whether $\lambda$ is feasible}
\BlankLine
%\tcc{All indices below are understood modulo $n$.}
$count\leftarrow 0$\;
\For{\em each stem $P$ of the stem tree $T_c$}
{
  $sup(P)\leftarrow \infty$, $dem(P)\leftarrow sup(P)-1$\;
}
\For{\em each stem $P$ in the post-order traversal of $T_c$}
{
    \For{\em each child $P'$ of $P$}
	{
     $sup(P)=\min\{sup(P),sup(P')\}$,
     $dem(P)=\min\{dem(P),dem(P')\}$\;
	}
    Increase $count$ by the number of twigs of $P$\;
    Let $V$ be the set of the uncovered vertices of $P$, and $v_1',v_2',\ldots,v_t'$ are the backbone vertices of $V$\;
    \eIf{$sup(P)\leq dem(P)$}
    {
      Let $q$ be the center in a child stem of $P$ that gives the value $sup(P)$      \;
      Do binary search to find the largest index $i\in [1,t]$ such that $q$ can cover all vertices of $V[1,i]$\;
      \label{ln:binary10}
      \eIf{$i=t$}
      {
        $sup(P)=\min\{sup(P)+d(v_1,v_m),d(b_b,v_m)\} $\tcc*{$b_b$ is the bud maintained for $P$}
      }
      {
         $i++$, $count=count+ncen(i)$\;
         \eIf{$w(v(i))\cdot d(v(i),v_m)< \lambda$}
          {
            $dem(P)=\lambda/w(v(i))-d(v(i),v_m)$, $sup(P)=d(b_b,v_m)$\;
          }
          {
             $count++$, $sup(P)=\min\{d(q(i),v_m),d(b_b,v_m)\}$, $dem(P)=\infty$\;
          }
      }
    }
    {
      \eIf(\tcc*[f]{$b_a$ is the bud maintained for $P$}){$dem(P)\geq d(v_1,b_a)$}
      {
         $count=count+ncen(1)$\;
         \eIf{$w(v(1))\cdot d(v(1),v_m)< \lambda$}
         {
           $dem(P)=\lambda/w(v(1))-d(v(1),v_m)$, $sup(P)=d(b_b,v_m)$\;
         }
         {
           $count++$, $sup(P)=\min\{d(q(1),v_m),d(b_b,v_m)\}$, $dem(P)=\infty$\;
         }
      }
      {
         Do binary search to find the largest $i\in [1,t]$ such that there exists a center $q$ at the backbone vertex of $P$ to cover all vertices of $V[1,i]$ with $d(v_1,q)\leq dem(P)$\;
         \label{ln:binary20}
         \eIf{$i=t$}
         {
           \eIf{\em $w(v(1))\cdot d(v(1),v_m)< \lambda$ and $dem(P)> d(v_1,v_m)$}
           {
              $dem(P)=\min\{dem(P)-d(v_1,v_m), \lambda/w(v(1))-d(v(1),v_m)\}$,  $sup(P)=d(b_b,v_m)$\;
           }
           {
              Do binary search to find the largest $j\in [1,m]$ such that $dem(P)\geq d(v_1,v_j)$\; \label{ln:binary30}
              $count++$,
              $\delta=\max\{d(v_j,v_m),d(q(1),v_m)\}$, $sup(P)=\min\{\delta, d(b_b,v_m)\}$, $dem(P)=\infty$\;
           }
         }
         {
           $i++$, $count=count+1+ncen(i)$\;
           \eIf{$w(v(i))\cdot d(v(i),v_m)< \lambda$}
           {
             $dem(P)=\lambda/w(v(i))-d(v(i),v_m)$, $sup(P)=d(b_b,v_m)$\;
           }
           {
             $count++$, $sup(P)=\min\{d(q(i),v_m),d(b_b,v_m)\}$, $dem(P)=\infty$\;
           }
         }
      }
    }

}
\If{$sup(P_{\gamma})>dem(P_{\gamma})$}
{
        $count++$\;
}
Return true if and only if $count\leq k$\;
%\eIf{$count\leq k$}
%{
%        return true\tcc*{$\gamma$ is feasible}
%}
%{
%        return false\tcc*{$\gamma$ is not feasible}
%}
\end{algorithm}

First of all, we can still implement Line~\ref{ln:binary10} in $O(\log^2r)$ time by exactly the same algorithm in Lemma~\ref{lem:140}. For the binary search in Line~\ref{ln:binary20}, since now we have a constraint that $q$ must be at a backbone vertex, we need to modify the algorithm in Lemma~\ref{lem:150}, as follows.

We show that given any index $i\in [1,t]$, we can determine in $O(\log^2 r)$ time whether there is a center at a backbone vertex of $P$ that can cover all vertices of $V[1,i]$ with $d(v_1,q)\leq dem(P)$. We first compute the center $q$ and its optimal objective value $\alpha'(1,i)$ for the discrete one-center problem on the vertices of $V[1,i]$, which can be done in $O(\log^2 r)$ time as shown the proof of Lemma~\ref{lem:160}. If $\alpha'(1,i)>\lambda$, then the answer is no. Otherwise, if $d(v_1,q)\leq dem(P)$, then the answer is yes. If $d(v_1,q)>dem(P)$, then let $j$ be the largest index of $[1,m]$ such that $dem(P)\geq d(v_1,v_j)$, and $j$ can be found in $O(\log r)$ time by binary search on the backbone vertices of $P$. Let $l$ be the vertical line of $\bbR^2$ whose $x$-coordinate is equal to $d(v_1,v_j)$. By a line-constrained 2D sublist LP query, we compute the lowest point $p'$ on $l$ in the common intersection of the upper half-planes defined by the vertices of $V[1,i]$. The answer is yes if and only if the $y$-coordinate of $p'$ is at most $\lambda$. Hence, the time to determine the answer to the above question is $O(\log^2 r)$. Therefore, the time for implementing Line~\ref{ln:binary20} is $O(\log^3 r)$.

In addition, it is easy to see that the time of the binary search in Line~\ref{ln:binary30} is $O(\log r)$.

Therefore, processing $P$ takes $O(\log^3 r)$ time, and the total time of \dftestnew\ is $O(n/r \log^3 r)$, the same as before.

\subsubsection{Phase 2}

This phase is the similar as before with the following changes. First, we use \dftestnew\ to replace \ftestnew. Second, we use the new post-processing procedure. Third, instead of using the line arrangement searching technique, we use \msearch. Specifically, in the pseudocode of Algorithm~\ref{algo:kcenter}, we replace Lines~\ref{ln:lineset10} and \ref{ln:vertex10} (and also Lines~\ref{ln:lineset20} and \ref{ln:vertex20} ) by the following. For each leaf-stem of $S$, we form the matrices for $P$ in the way discussed in Section~\ref{sec:dstem}, and let $\calM$ denote the set of matrices for all leaf-stems of $S$. Then, we call \msearch\ on $\calM$ with stopping count $c=0$ and \dftestnew.
%Similarly, we also replace Lines~\ref{ln:lineset20} and \ref{ln:vertex20} by the above algorithm.

The running time of all three phases is still $O(n\log n)$, as shown in Theorem~\ref{theo:dkcenter}.

\begin{theorem}\label{theo:dkcenter}
The discrete $k$-center problem for $T$ can be solved in $O(n\log n)$ time.
\end{theorem}
\begin{proof}
The analysis is similar to that in Theorem~\ref{theo:kcenter}, we briefly discuss it below. Since Phase 0 and Phase 1 run in $O(n\log n)$ time, we only discuss Phase 2.

Again, the number of iterations of the while loop is $O(\log n)$. Hence, there are $O(\log n)$ calls to \msearch. Each call to \msearch\ produces $O(\log n)$ values for feasibility tests. Therefore, the total number of feasibility tests is $O(\log^2 n)$. With \dftestnew, the total time on feasibility tests is $O(n/r\cdot \log^3 r\cdot \log^2 n)=O(n\log n)$. In each iteration, let $n'$ denote the total number of backbone vertices of all leaf-stems. According to our discussion in Section~\ref{sec:dstem}, the call to \msearch\ takes $O(n'\log n')$ time (excluding the time for feasibility tests) since each matrix element of $\calM$ can be obtained in $O(1)$ time. As the total sum of all such $n'$ is $O(n)$ in Phase 2, the overall time of \msearch\ in Phase 2 is $O(n\log n)$. Also, the overall time for the post-processing procedure in Phase 2 is $O(n)$.
Therefore, the total time of Phase 2 is $O(n\log n)$. This proves the theorem.
\qed
\end{proof}

%\section{Concluding Remarks}

\bibliography{reference}
\bibliographystyle{plain}

\end{document}